\documentclass[sigconf,screen]{acmart}

\newif\ifanonymous
\anonymousfalse
\relax

\AtBeginDocument{}

\usepackage[utf8]{inputenc} 
\usepackage[T1]{fontenc}    
\usepackage{xcolor}         
\definecolor{linkdarkblue}{rgb}{0, 0.08, 0.45}    
\usepackage{url}            
\usepackage{array}          
\usepackage{colortbl}       
\usepackage{makecell}       
\usepackage{tabularx}       
\usepackage{bigstrut}       
\usepackage{multirow}       
\usepackage{colortbl}       
\usepackage{nicefrac}       
\usepackage{microtype}      
\usepackage{graphicx}       
\usepackage{bm}             
\usepackage[capitalise,noabbrev]{cleveref} 
\usepackage{subcaption}     
\usepackage{booktabs}       
\usepackage{wrapfig}        
\usepackage{pifont}         
\usepackage{algorithm}      
\usepackage{mathtools}      
\usepackage{enumitem}       
\usepackage{pifont}         
\usepackage{algorithm}      
\usepackage{algorithmicx}   
\usepackage{algpseudocode}  
\usepackage{thm-restate}    
\usepackage{fontawesome}    
\usepackage{comment}        
\usepackage[normalem]{ulem} 
\usepackage[most]{tcolorbox}
\usepackage{xspace}         
\usepackage{amsmath,amsfonts,bm}

\def\eqref#1{equation~\ref{#1}}

\def\1{\bm{1}}

\DeclareMathAlphabet{\mathsfit}{\encodingdefault}{\sfdefault}{m}{sl}
\SetMathAlphabet{\mathsfit}{bold}{\encodingdefault}{\sfdefault}{bx}{n}

\allowdisplaybreaks

\newcommand{\noindentparnolinenospace}[1]{\noindent\textbf{#1} }

\newcommand{\noindentparnoline}[1]{\addvspace{0.50\baselineskip}\noindent\textbf{#1} }

\newcommand{\eqnref}[1]{(\ref{#1})}

\numberwithin{equation}{section}    
\numberwithin{algorithm}{section}   

\theoremstyle{plain}

\newtheorem*{theorem*}{Theorem}

\newtheorem*{lemma*}{Lemma}

\newtheorem*{axiom*}{Axiom}

\newtheorem*{conjecture*}{Conjecture}

\newtheorem*{assumption*}{Assumption}

\newtheorem*{proposition*}{Proposition}

\newtheorem*{corollary*}{Corollary}

\theoremstyle{definition}

\newtheorem*{definition*}{Definition}

\newtheorem*{example*}{Example}

\newtheorem*{problem*}{Problem}

\newtheorem*{question*}{Question}

\newtheorem*{exercise*}{Exercise}

\newtheorem*{remark*}{Remark}

\definecolor{darkgreen}{rgb}{0.0, 0.5, 0.0}
\definecolor{skyblue}{rgb}{0.53, 0.81, 0.98}
\definecolor{beamsearchblue}{rgb}{0.42, 0.56, 0.75}     
\definecolor{beamsearchyellow}{rgb}{0.84, 0.71, 0.34}   
\definecolor{beamsearchred}{rgb}{0.72, 0.33, 0.31}      
\definecolor{darkred}{rgb}{0.6, 0, 0}
\definecolor{pruning_red}{rgb}{0.64, 0.12, 0.21}        
\definecolor{pruning_gray}{rgb}{0.6, 0.6, 0.6}          
\definecolor{sensitivity_blue}{rgb}{0.30, 0.45, 0.69}   
\definecolor{sensitivity_orange}{rgb}{0.87, 0.52, 0.32} 

\newcolumntype{Y}{>{\centering\arraybackslash}X}

\newcommand{\ie}{\text{i.e., }}         
\newcommand{\eg}{\text{e.g., }}         
\newcommand{\wrt}{\text{w.r.t. }}       
\newcommand{\cf}{\text{cf. }}           
\newcommand{\st}{\text{s.t. }}          

\newtcolorbox{statementbox}[1][]
{
    enhanced,
    colback=blue!5!white,
    colframe=blue!60!black,
    coltitle=white,
    fonttitle=\bfseries,
    boxrule=0pt,
    bottomrule=1.5pt,
    arc=0pt,
    auto outer arc,
    left=8pt, right=8pt, top=5pt, bottom=5pt,
    toptitle=2pt, bottomtitle=2pt,
    before skip=8pt, after skip=8pt,
    breakable,
    #1
}

\definecolor{bocchipink}{HTML}{F3A6B2}
\definecolor{nijikayellow}{HTML}{F8DD74}
\definecolor{ryoblue}{HTML}{6680AE}
\definecolor{kitared}{HTML}{DD6D63}

\tcbset{
    kessokustyle/.style={
        enhanced, breakable,
        colback=#1!8!white,
        colframe=#1!90!black,
        boxrule=0pt, leftrule=3pt, arc=0pt, auto outer arc,
        left=8pt, right=8pt, top=5pt, bottom=5pt,
        before skip=8pt, after skip=8pt,
    },
}

\newcommand{\mymethod}{SPRINT\xspace}

\copyrightyear{2026}
\acmYear{2026}
\setcopyright{cc}
\setcctype{by-nc-nd}
\acmConference[KDD 2026] {Proceedings of the 32nd ACM SIGKDD Conference on Knowledge Discovery and Data Mining V.2}{August 9--13, 2026}{Jeju Island, Republic of Korea.}
\acmBooktitle{Proceedings of the 32nd ACM SIGKDD Conference on Knowledge Discovery and Data Mining V.2 (KDD 2026), August 9--13, 2026, Jeju Island, Republic of Korea}
\acmISBN{979-8-4007-2259-2/2026/08}
\acmDOI{10.1145/3770855.3818185}

\makeatletter
\renewcommand{\@fnsymbol}[1]{%
  \ensuremath{\ifcase#1\or \dagger\or \ast\or \ddagger\or
  \mathsection\or \mathparagraph\or \|\or **\or \dagger\dagger
  \or \ddagger\ddagger \else\@ctrerr\fi}}
\makeatother

\begin{document}

\title[The Pitfall of Scaling Up: Uncovering and Mitigating Popularity Bias Amplification in \\Scaling Transformer-based Recommenders]{
The Pitfall of Scaling Up: Uncovering and Mitigating Popularity Bias Amplification in Scaling Transformer-based Recommenders
}


\settopmatter{authorsperrow=4}

\author{Weiqin Yang}
    \orcid{0000-0002-5750-5515}
    \affiliation{
        \institution{Zhejiang University}
        \city{Hangzhou}
        \country{China}
    }
    \email{tinysnow@zju.edu.cn}
\authornote{Equal contribution.}
\authornotemark[3]
\authornotemark[4]

\author{Yue Pan}
    \orcid{0009-0009-3726-5145}
    \affiliation{
        \institution{Zhejiang University}
        \city{Hangzhou}
        \country{China}
    }
    \email{3220100912@zju.edu.cn}
\authornotemark[1]
\authornotemark[3]
\authornotemark[4]

\author{Chongming Gao}
    \orcid{0000-0002-5187-9196}
    \affiliation{
        \institution{University of Science and Technology of China}
        \city{Hefei}
        \country{China}
    }
    \email{chongming.gao@gmail.com}

\author{Sheng Zhou}
    \orcid{0000-0003-3645-1041}
    \affiliation{
        \institution{Zhejiang University}
        \city{Hangzhou}
        \country{China}
    }
    \email{zhousheng_zju@zju.edu.cn}

\author{Xiang Wang}
    \orcid{0000-0002-6148-6329}
    \affiliation{
        \institution{University of Science and Technology of China}
        \city{Hefei}
        \country{China}
    }
    \email{xiangwang1223@gmail.com}

\author{Can Wang}
    \orcid{0000-0002-5890-4307}
    \affiliation{
        \institution{Zhejiang University}
        \city{Hangzhou}
        \country{China}
    }
    \email{wcan@zju.edu.cn}
\authornotemark[3]
\authornotemark[4]

\author{Jiawei Chen}
    \orcid{0000-0002-4752-2629}
    \affiliation{
        \institution{Zhejiang University}
        \city{Hangzhou}
        \country{China}
    }
    \email{sleepyhunt@zju.edu.cn}
\authornote{Corresponding author.}
\authornote{State Key Laboratory of Blockchain and Data Security, Zhejiang University.}
\authornote{College of Computer Science and Technology, Zhejiang University.}
\authornote{Hangzhou High-Tech Zone (Binjiang) Institute of Blockchain and Data Security.}

\ifanonymous
\else
    \renewcommand{\shortauthors}{Weiqin Yang et al.}
\fi



\begin{abstract}
    We identify a critical pitfall in scaling transformer-based sequential recommenders: while increasing model size improves recommendation accuracy, it simultaneously amplifies \emph{popularity bias}. This bias drives systems to over-recommend popular items at the expense of niche ones, which not only undermines fairness but also degrades the broader ecosystem by reinforcing the Matthew effect and filter bubbles. Consequently, this bias amplification emerges as a fundamental obstacle to sustainable model scaling.

    Through comprehensive theoretical and empirical analyses, we uncover the root cause of this amplification. Our findings reveal that as model depth increases, the two core components of the transformer architecture, \ie attention aggregation and feed-forward projections, synergistically induce severe \emph{spectral collapse} in model predictions, which directly translates to the amplification of popularity bias. To address this challenge, we propose \textbf{\mymethod} (\textbf{\uline{S}}calable \textbf{\uline{P}}opularity \textbf{\uline{R}}egularization \textbf{\uline{IN}} \textbf{\uline{T}}ransformers), which mitigates spectral collapse during scaling by constraining (i) the maximum column-sums of the attention score matrices and (ii) the spectral norms of the feed-forward parameters. Extensive experiments demonstrate that \mymethod significantly improves both accuracy and long-tail fairness. Crucially, it yields more favorable scaling behaviors when expanding model sizes from 0.05M to 0.34B parameters. The code is available at \url{https://github.com/Tiny-Snow/GenRec}.
\end{abstract}





\maketitle

\begingroup\small\noindent\raggedright\textbf{Resource Availability:}\\
The source code used in this paper is available at \url{https://github.com/Tiny-Snow/GenRec}, and an archive copy is maintained at \url{https://doi.org/10.5281/zenodo.20400259}. The datasets adopted in this paper are preprocessed and can be accessed at \url{https://github.com/Tiny-Snow/GenRec-Datasets}, and an archive copy is maintained at \url{https://doi.org/10.5281/zenodo.20400316}.
\endgroup



\section{Introduction} \label{sec:introduction}

The discovery of scaling laws~\citep{kaplan2020scaling} has fundamentally reshaped the landscape of modern machine learning. These laws characterize a remarkable empirical regularity: model performance improves predictably (often following a power law) as the number of parameters or the volume of data increases~\citep{hoffmann2022training,wei2022emergent}. Guided by these observations, transformer-based foundation models have been scaled to billions or even trillions of parameters~\citep{liu2024deepseek,yang2025qwen3}, unlocking unprecedented proficiency in contextual understanding~\citep{brown2020language,hao2025rethinking}, complex reasoning~\citep{wei2022chain,wang2025scheduling,chen2025arrows}, and agency~\citep{hao2026recreate,hao2026evolveteamcollaborativeselfevolution}.

Inspired by this success, researchers have increasingly explored scaling laws in recommender systems, particularly in sequential recommendation~\citep{kang2018self,sun2019bert4rec,wang2025msl,wang2025llm4dsr,shi2024large,zhang2025reinforced}. Recent studies have observed a similar trend: scaling up the parameters of transformer-based sequential recommenders leads to consistent improvements in recommendation accuracy, which has also been validated in large-scale industrial recommender systems~\citep{guo2024scaling,zhang2024wukong,zhang2024scaling,zivic2024scaling,chen2026beyond,chen2024sigformer,chen2025rankformer,wang2026trustworthy}.

\begin{figure}[t]
    \centering
    \includegraphics[width=\columnwidth]{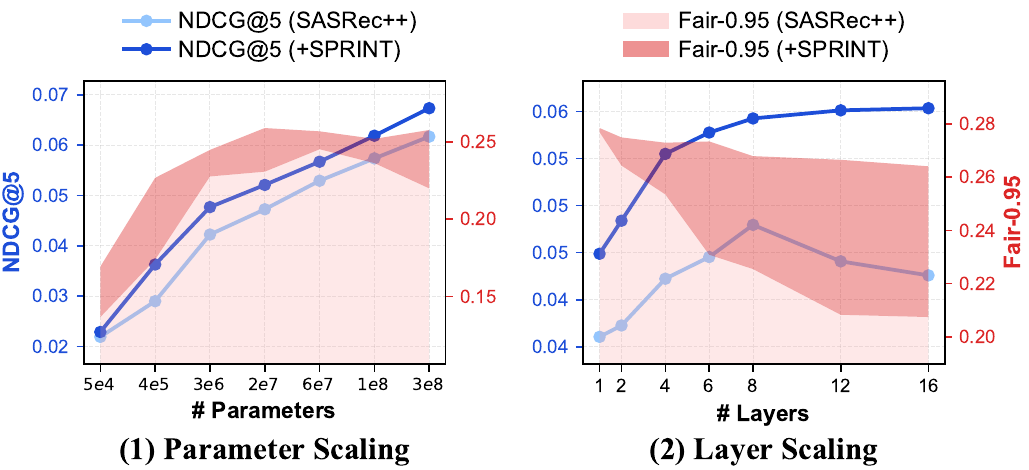}
    \caption{
        Scaling laws of accuracy (NDCG@5) and fairness (Fair-0.95) on transformer-based sequential recommenders, across (1) different number of non-embedding parameters (0.05M to 0.34B) and (2) different transformer depths (1 to 16 layers). Here we present results on the MovieLens-20M dataset~\citep{harper2015movielens} with the advanced SASRec++ backbone~\citep{zhang2024scaling,guo2024scaling}, and the fairness metric denotes the ratio of bottom-95\% long-tail items exposed in top-1 recommendations. Notably, as model size increases, especially by deepening layers, fairness tends to degrade due to amplified popularity bias. Our proposed \mymethod effectively mitigates this bias, enabling superior scaling performance in both accuracy and fairness.
    }
    \Description{Movielens-20M scaling laws.}
    \label{fig:ml-20m-scaling-laws}
\end{figure}

\noindentparnoline{Challenge.}
However, we uncover a \emph{critical pitfall} of scaling sequential recommenders: as model size increases, especially when increasing transformer depth, popularity bias is simultaneously amplified, as shown in \cref{fig:ml-20m-scaling-laws}. Popularity bias~\citep{chen2023bias,chen2021autodebias}, a notorious issue in recommender systems, denotes the tendency to over-recommend popular items at the expense of niche ones. This bias adversely affects both accuracy and fairness and can further degrade the recommendation ecosystem by reinforcing the Matthew effect and intensifying filter bubbles through the user-system feedback loop~\citep{lin2025recommendation}. The severity of these consequences poses a fundamental challenge to achieving sustainable scaling in recommendation.

\noindentparnoline{Analysis.}
To tackle this challenge, we first conduct theoretical and empirical analyses to pinpoint the root cause behind the amplification of popularity bias during model scaling. Prior work by \citet{lin2025recommendation} indicates that item popularity is primarily encoded within the principal spectral component of the model's prediction score matrix, which is weighted by the largest singular value (\ie the spectral norm) in the singular value decomposition (\cf \cref{lemma:popularity-bias-spectral-norm}). Based on this insight, we further empirically observe that increasing transformer depth exacerbates \emph{spectral collapse}~\citep{jing2021understanding,guo2023embedding,zhang2023mitigating,ohsaka2023curse}: the largest singular value grows disproportionately dominant relative to the rest of the singular values, thereby magnifying the influence of the popularity-related principal component on final recommendations, and consequently amplifying popularity bias as model scales up (\cf \cref{fig:ml-20m-scaling-laws}).

Diving deeper into the architectural mechanisms, we attribute the aforementioned depth-induced spectral collapse to the two core components of transformer-based sequential recommenders: \emph{attention aggregation} and \emph{feed-forward projections}. Specifically, the attention mechanism can over-emphasize popular items as they tend to accumulate disproportionately high attention weights. Our theoretical analysis further shows that stacking layers with such popularity-priority attention intensifies spectral collapse by increasing the spectral norm exponentially with depth (\cf \cref{theorem:spectral-norm-attention-aggregation}). Moreover, the gradient dynamics of stacked feed-forward projections can make gradient-based optimization less effective in learning features within the long-tail singular space, further exacerbating spectral collapse as the network deepens (\cf \cref{lemma:spectral-norm-feed-forward-projection}). Together, these two mechanisms synergistically amplify popularity bias and impede the scalability of sequential recommenders.

\noindentparnoline{Our Method.}
Given the adverse impact of popularity bias on the fairness and sustainability of recommendation systems, developing effective debiasing strategies is crucial. While numerous debiasing methods have been proposed~\citep{wang2021deconfounded,zhang2022incorporating,zhang2023mitigating,zhang2023model,zhang2023invariant,ning2024debiasing,wang2024causally,lin2025recommendation}, most primarily target traditional collaborative filtering scenarios, overlooking the unique characteristics of sequential recommendation and the implications of model architectures on scaling laws. Consequently, these approaches are less effective in addressing the popularity bias amplification observed during scaling.

To bridge this gap, we propose \textbf{\mymethod}, a popularity regularization method tailored to transformer-based sequential recommenders. Motivated by our prior theoretical analysis, \mymethod introduces two key regularizers: (i) \textit{attention regularization}, which constrains the maximum column-sums of the attention score matrices, limiting the excessive accumulation of popularity signals during attention aggregation, thus mitigating popularity bias amplification; and (ii) \textit{feed-forward regularization}, which bounds the spectral norms of the feed-forward projection weights, thereby controlling spectral collapse as depth increases. By regularizing these core components in each layer, \mymethod can natively alleviate spectral collapse when training deeper transformers.

\mymethod offers three appealing properties: (i) We prove that \mymethod upper-bounds the spectral norm of the prediction score matrix, providing theoretical guarantees for mitigating popularity bias amplification (\cf \cref{theorem:spectral-norm-bound}). (ii) Unlike existing debiasing approaches, \mymethod directly targets the architectural sources of this amplification by regularizing the specific transformer components. (iii) \mymethod is computationally efficient, introducing negligible training overhead. These properties make \mymethod a principled and effective debiasing solution for scaling sequential recommenders.

We evaluate \mymethod on six benchmark datasets using two strong sequential recommendation backbones: SASRec++~\citep{zhang2024scaling,guo2024scaling} and HSTU~\citep{zhai2024actions}. Experimental results demonstrate that \mymethod significantly improves both accuracy (avg. +15.70\% in NDCG and HR) and long-tail exposure (avg. +7.12\%). To assess scalability, we test model sizes ranging from 0.05M to 0.34B non-embedding parameters (a 6912$\times$ scale-up). The results show that \mymethod achieves better accuracy scaling while maintaining a more balanced popularity distribution (\cf \cref{fig:ml-20m-scaling-laws}). We further extend \mymethod to the generative recommendation scenario and observe consistent gains on the state-of-the-art backbones TIGER~\citep{rajput2023recommender} and LETTER~\citep{wang2024learnable}.

\begin{statementbox}[title={Our Contributions}]
\begin{itemize}[topsep=0pt,leftmargin=3pt,itemsep=0pt]
\item We identify a critical scaling pitfall in sequential recommendation: while accuracy improves with model scaling, popularity bias is simultaneously amplified.
\item We reveal that attention aggregation and feed-forward projections in transformer synergistically exacerbate the spectral collapse of model predictions, thereby intensifying popularity bias as depth grows.
\item We propose \mymethod, a principled and effective popularity regularization method that constrains the maximum column-sums of the attention matrices and the spectral norms of feed-forward parameters to mitigate this spectral collapse.
\item Extensive experiments demonstrate that \mymethod significantly improves both accuracy and fairness (\ie long-tail exposure), yielding more favorable scaling behaviors.
\end{itemize}
\end{statementbox}


\section{Preliminaries} \label{sec:preliminaries}

In this section, we formulate the sequential recommendation task in \cref{sec:preliminaries:task-formulation}, introduce the popularity bias in \cref{sec:preliminaries:popularity-bias}, and present the transformer-based model architectures in \cref{sec:preliminaries:model-formulation}.


\subsection{Task Formulation} \label{sec:preliminaries:task-formulation}

In this work, we focus on sequential recommendation~\citep{wang2026does,cui2024distillation,cui2025hatllm,cui2026field,cui2026spectran}, a widely adopted paradigm in real-world recommender systems. Formally, given a set of users $\mathcal{U}$ and a set of items $\mathcal{I}$, each user $u \in \mathcal{U}$ interacts with a sequence of items $(i_{u, 1}, i_{u, 2}, \cdots, i_{u, n_u})$ in chronological order. Here, $i_{u, t} \in \mathcal{I}$ denotes the $t$-th item the user interacts with, and $n_u$ represents the length of the interaction sequence. Given this historical sequence as input, the goal of a sequential recommendation model is to predict the next item $i_{u, n_u + 1}$ that the user is most likely to interact with.



\subsection{Popularity Bias} \label{sec:preliminaries:popularity-bias}

Recommender systems often exhibit a long-tail distribution of item popularity, where a small fraction of items (\ie \emph{popular items}) receive a disproportionately large number of interactions, while the majority of items (\ie \emph{long-tail items}) are rarely interacted with~\citep{chen2023bias,lin2025recommendation}. For instance, in the ML-1M dataset~\citep{harper2015movielens}, the top-20\% of items (ranked by popularity) account for 62.6\% of all interactions, while the bottom-40\% account for only 5.4\%. Such a skewed data distribution poses significant challenges for training effective recommendation models, as these models often amplify this skewness by over-recommending popular items. For example, when employing the state-of-the-art HSTU model~\citep{zhai2024actions} on ML-1M, we observe that over 80\% of the recommendations are dominated by the top-20\% popular items, whereas the bottom-40\% long-tail items are recommended less than 1\% of the time. This notorious phenomenon, known as \emph{popularity bias}, not only hinders recommendation accuracy but also leads to information cocoons, inequitable exposure, and reduced user satisfaction~\citep{zhu2021popularity,klimashevskaia2024survey}.

To quantify the extent of popularity bias in recommendation systems, we additionally measure a \emph{fairness metric} Fair-$\alpha$, defined as the ratio of bottom-$\alpha$ long-tail items in the final recommendations, which is also referred to as the long-tail exposure. Following prior work~\citep{lin2025recommendation}, we evaluate Fair-0.8 by treating the bottom 80\% of items as long-tail items across most datasets (\cf \cref{tab:main-results-5}). For the ML-20M dataset, due to its massive scale and particularly skewed popularity distribution, we instead evaluate Fair-0.95 (\cf \cref{fig:ml-20m-scaling-laws}). In the main text, we primarily evaluate fairness in top-1 recommendations, while the similar trends can be observed in top-5 and top-10 recommendations (\cf \cref{tab:fairness-results}).



\subsection{Model Architecture} \label{sec:preliminaries:model-formulation}

Motivated by the widespread success of the transformer architecture~\citep{vaswani2017attention}, transformer-based models have gained significant traction and achieved state-of-the-art performance in sequential recommendation~\citep{kang2018self,sun2019bert4rec}. In the following, we present a generic formulation of these models. For ease of exposition, we detail only the single-head attention mechanism and standard two-layer feed-forward networks, purposely omitting auxiliary components like residual connections and layer normalization. Crucially, our theoretical framework naturally extends to incorporate these components; the complete analyses are deferred to \cref{app:theoretical-proofs:proof-of-spectral-norm-bound}.

\noindentparnoline{Embedding.}
In the initial embedding layer, the input sequence of each user $u$ is transformed into a sequence of embedding vectors $\mathbf{X}_u^{(0)} = [\mathbf{x}_{u, 1}^{(0)}; \mathbf{x}_{u, 2}^{(0)}; \cdots; \mathbf{x}_{u, n_u}^{(0)}] \in \mathbb{R}^{n_u \times d}$, where $\mathbf{x}_{u, t}^{(0)} \in \mathbb{R}^{1 \times d}$ is the $d$-dimensional embedding of item $i_{u, t}$. Stacking the embeddings of all users yields the initial representation $\mathbf{X}^{(0)} = [\mathbf{X}_u^{(0)}]_{u \in \mathcal{U}} \in \mathbb{R}^{N \times d}$, where $N = \sum_{u \in \mathcal{U}} n_u = |\mathcal{U}| \cdot \bar{n}$ is the total number of interactions, and $\bar{n}$ is the average user sequence length. This zeroth-layer representation $\mathbf{X}^{(0)}$ is then fed into a stack of $L$ transformer layers to capture sequential dependencies and contextual information.

\noindentparnoline{Attention Aggregation.}
For each layer $\ell = 1, 2, \cdots, L$, the input representation $\mathbf{X}^{(\ell-1)}$ from the $(\ell-1)$-th layer is updated to the \emph{aggregated representation} $\mathbf{A}^{(\ell)} \in \mathbb{R}^{N \times d}$ via self-attention:
\begin{equation} \label{eq:attention-aggregation}
    \mathbf{A}^{(\ell)} = \mathbf{S}^{(\ell)} \mathbf{X}^{(\ell-1)} \mathbf{W}_V^{(\ell)} \mathbf{W}_O^{(\ell)},
\end{equation}
where $\mathbf{S}^{(\ell)} = \operatorname{diag}(\mathbf{S}_u^{(\ell)})_{u \in \mathcal{U}} \in \mathbb{R}^{N \times N}$ is the block-diagonal attention score matrix. Each block $\mathbf{S}_u^{(\ell)} \in \mathbb{R}^{n_u \times n_u}$ represents the attention scores for user $u$ at layer $\ell$. Additionally, $\mathbf{W}_V^{(\ell)}, \mathbf{W}_O^{(\ell)} \in \mathbb{R}^{d \times d}$ are the value and output projection matrices, respectively. In practice, attention aggregation is computed independently for each user sequence (\ie within the corresponding block). Here, we present the formulation in a unified matrix form for ease of theoretical analysis.

\noindentparnoline{Feed-forward Projection.} 
After attention, the aggregated representation $\mathbf{A}^{(\ell)}$ is further transformed into the output representation $\mathbf{X}^{(\ell)}$ via a feed-forward network (FFN), typically consisting of two linear transformations and a non-linear activation function $\mu(\cdot)$:
\begin{equation} \label{eq:feed-forward-network}
    \mathbf{X}^{(\ell)} = \mu(\mathbf{A}^{(\ell)} \mathbf{W}_1^{(\ell)}) \mathbf{W}_2^{(\ell)},
\end{equation}
where $\mathbf{W}_1^{(\ell)} \in \mathbb{R}^{d \times d_{f}}$ and $\mathbf{W}_2^{(\ell)} \in \mathbb{R}^{d_{f} \times d}$ are the weight matrices of the two linear layers, and $d_{f}$ is the intermediate dimension. The feed-forward projection serves to further enhance model expressiveness by applying non-linear transformations to representations.

\noindentparnoline{Prediction.}
Finally, the output representation $\mathbf{X}^{(L)}$ from the last layer is used to generate prediction scores over all items based on its inner product with the item embedding matrix $\mathbf{E} \in \mathbb{R}^{|\mathcal{I}| \times d}$, \ie $\hat{\mathbf{Y}} = \mathbf{X}^{(L)} \mathbf{E}^\top \in \mathbb{R}^{N \times |\mathcal{I}|}$. The items with the highest prediction scores are retrieved as final recommendations.




\section{Pitfalls of Scaling Laws} \label{sec:pitfalls-of-scaling-laws}

In this section, we identify a critical pitfall in scaling sequential recommendation models: \emph{while scaling up model size generally enhances performance, it simultaneously amplifies popularity bias}. We first provide empirical evidence in \cref{sec:pitfalls-of-scaling-laws:empirical-evidences}, then conduct empirical and theoretical analyses of the root causes in \cref{sec:pitfalls-of-scaling-laws:analyses}.


\subsection{Popularity Bias Amplification} \label{sec:pitfalls-of-scaling-laws:empirical-evidences}

Motivated by recent efforts to scale sequential recommendation models~\citep{guo2024scaling,zhang2024wukong,zhang2024scaling,zivic2024scaling}, we conduct an empirical study to investigate the impact of model size on both recommendation accuracy and fairness. As shown in \cref{fig:ml-20m-scaling-laws}.(1), our results reveal an intriguing pattern: while overall recommendation accuracy improves with larger model sizes, recommendation fairness exhibits a non-monotonic trajectory. Specifically, popularity bias decreases as model size grows to a moderate scale (\eg 18.9M), likely due to an enhanced capacity to align with user preferences. However, when further scaling to large sizes (\eg 0.13B and 0.34B), popularity bias significantly worsens, suggesting an increased tendency to over-recommend popular items and underutilize long-tail features. More crucially, this popularity bias is continuously amplified as model depth increases, eventually becoming a bottleneck to further scaling. In \cref{fig:ml-20m-scaling-laws}.(2), we observe accuracy saturation, and in some cases even degradation, when model depth is increased (\eg over 8 layers) without additional debiasing mechanisms. Overall, these findings reveal a critical pitfall in scaling sequential recommenders:

\begin{statementbox}[title={Observation: Popularity Bias Pitfall in Scaling}]
When scaling model size, particularly by increasing model depth, popularity bias is simultaneously amplified and can become a bottleneck to further scaling.
\end{statementbox}



\subsection{Analyses on the Root Causes} \label{sec:pitfalls-of-scaling-laws:analyses}

Our observations in \cref{sec:pitfalls-of-scaling-laws:empirical-evidences} have identified the critical pitfall in scaling transformer-based sequential recommenders, prompting a natural question: \emph{what are the root causes of popularity bias amplification during model scaling?}

\noindentparnoline{Spectral Norm and Popularity Bias.}
To answer this question, we first revisit the theoretical insights from \citet{lin2025recommendation}, which reveal an inherent connection between popularity bias and the spectral norm (\ie the largest singular value) of the predicted score matrix. For convenience, we denote the item popularity vector as $\mathbf{r}$ (where $\mathbf{r}_i$ is the number of interactions for item $i$) and the singular value decomposition (SVD) of the predicted score matrix as $\hat{\mathbf{Y}} = \sum_{k = 1}^{\min(N, |\mathcal{I}|)} \sigma_k(\hat{\mathbf{Y}}) \mathbf{p}_k \mathbf{q}_k^\top$, where $\sigma_k(\hat{\mathbf{Y}})$ is the $k$-th largest singular value, and $\mathbf{p}_k, \mathbf{q}_k$ are the corresponding left and right singular vectors, respectively. \citet{lin2025recommendation} prove the following result:

\begin{tcolorbox}[kessokustyle=nijikayellow]
\begin{restatable}{lemma}{RestatablePopularityBiasSpectralNorm} \label{lemma:popularity-bias-spectral-norm}
    \upshape
    {(\textit{Popularity Memorization}, Theorem 1 in~\citep{lin2025recommendation}).}
    If the item popularity $\mathbf{r}$ is sufficiently long-tailed (\ie following a power-law distribution), the principal singular vector $\mathbf{q}_1$ of the predicted score matrix $\hat{\mathbf{Y}}$ is highly correlated with the item popularity $\mathbf{r}$, \ie $\cos(\mathbf{q}_1, \mathbf{r}) \approx 1$.
\end{restatable}
\end{tcolorbox}


\noindent
\cref{lemma:popularity-bias-spectral-norm} indicates that the popularity bias is primarily captured in the principal spectral component of the predicted score matrix, which is also verified by our empirical results (\cf \cref{app:experimental-supplementary:supplementary-results}). Therefore, the spectral norm $\|\hat{\mathbf{Y}}\|_2 = \sigma_1(\hat{\mathbf{Y}})$, which dictates the magnitude of this principal component, can be interpreted as the strength of the popularity influence on final recommendations.

\begin{figure}[t]
    \centering
    \includegraphics[width=\columnwidth]{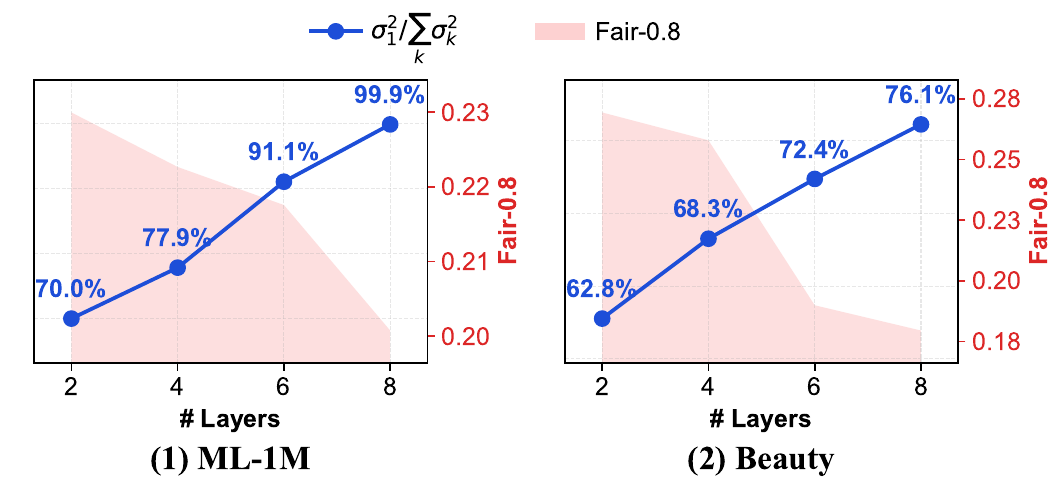}
    \caption{
        Illustration of increasingly severe spectral collapse and worsened fairness as the model deepens. The proportion of the largest singular value rises (\eg $\sigma_1^2/\sum_k \sigma_k^2$ reaching 99.9\% on the ML-1M dataset) as more layers are stacked, causing the popularity-related principal spectral component to exert a greater influence on the final recommendations, thereby degrading fairness (\ie Fair-0.8, the ratio of bottom-80\% long-tail items exposed in top-1 recommendation).
    }
    \Description{Spectral norm analysis of layer scaling.}
    \label{fig:analysis-layer-scaling-spectral-norm}
\end{figure}

\noindentparnoline{Spectral Collapse in Deep Transformers.}
Building on the above insight, we further analyze how transformer architectures contribute to \emph{spectral collapse} in model outputs. Here, spectral collapse refers to the phenomenon that the largest singular value of the prediction matrix dominates the singular value spectrum, \ie $\sigma_1$ becomes disproportionately large relative to the remaining singular values $\sigma_{k \geq 1}$. As a consequence, the principal spectral component, which is closely associated with item popularity, exerts an excessively strong influence on recommendation outcomes. Empirically, in \cref{fig:analysis-layer-scaling-spectral-norm}, we observe that spectral collapse consistently intensifies with depth, suggesting that deeper transformer stacks inherently exacerbate this effect and drive bias amplification.

We next provide a theoretical analysis to reveal the root causes of spectral collapse induced by transformer architectures. To facilitate intuitive understanding and theoretical derivations, we assume identity activations in the feed-forward networks. This simplification is common in theoretical studies of transformers, and our empirical results indicate that the resulting insights also hold under standard nonlinear activations (\cf \cref{app:theoretical-proofs:proof-of-spectral-norm-bound}). The output after $L$ transformer layers can be rewritten as follows:
\begin{equation} \label{eq:output-decomposition}
    \mathbf{X}^{(L)} = \underbrace{\left( \mathbf{S}^{(L)} \mathbf{S}^{(L - 1)} \cdots \mathbf{S}^{(1)} \right)}_{\text{\small attention aggregation}} \mathbf{X}^{(0)} \underbrace{\left( \mathbf{W}^{(1)} \mathbf{W}^{(2)} \cdots \mathbf{W}^{(L)} \right)}_{\text{\small feed-forward projection}},
\end{equation}
where $\mathbf{W}^{(\ell)}$ is the combined feed-forward weight product for each layer $\ell$, \eg $\mathbf{W}^{(\ell)} = \mathbf{W}_V^{(\ell)} \mathbf{W}_O^{(\ell)} \mathbf{W}_1^{(\ell)} \mathbf{W}_2^{(\ell)}$. This formulation elegantly highlights the distinct roles of the attention mechanisms (\emph{left}) and the feed-forward networks (\emph{right}) in aggregating and transforming the input representations, respectively. More importantly, \emph{these two components synergistically contribute to the spectral collapse of the final outputs}.

\begin{figure}[t]
    \centering
    \includegraphics[width=0.9\columnwidth]{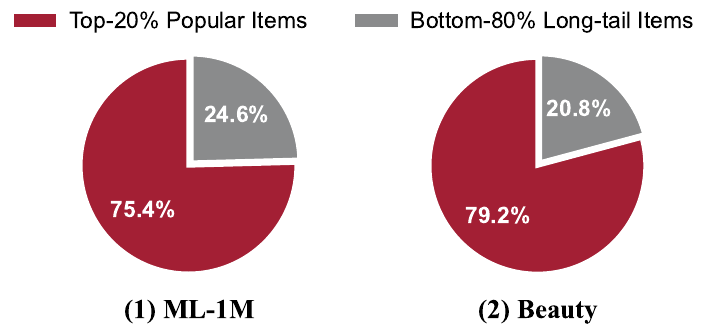}
    \caption{    
    Popular items tend to receive disproportionately higher attention scores compared to long-tail items. As illustrated, the top-20\% popular items account for over 75\% of the total attention weights, thereby dominating the aggregation process and amplifying popularity bias.
    }
    \Description{Attention score distribution analysis.}
    \label{fig:analysis-attention-scores}
\end{figure}

\noindentparnoline{Popularity-Prioritized Attention Aggregation.}
In attention mechanisms, the attention scores $\mathbf{S}^{(\ell)}$ dictate the overall influence of each item during the aggregation process. Our empirical results in \cref{fig:analysis-attention-scores} reveal that \emph{popular items tend to receive disproportionately higher attention scores compared to long-tail items}. Specifically, they account for over 75\% of the total attention weights, indicating their dominant contribution to the aggregated representations. More importantly, given the similar popularity-prioritized attention patterns across different layers, stacking multiple attention layers (\cf the left part of \cref{eq:output-decomposition}) may further amplify this effect, as established in the following theorem:

\begin{tcolorbox}[kessokustyle=bocchipink]
\begin{restatable}{theorem}{RestatableSpectralNormAttentionAggregation} \label{theorem:spectral-norm-attention-aggregation}
    \upshape
    {(\textit{Spectral Collapse in Attention Aggregation}).}
    Suppose that the attention score matrices $\mathbf{S}^{(\ell)}$ share similar singular vectors across layers up to an $\epsilon$-perturbation, \ie $\mathbf{S}^{(\ell)} = \mathbf{U} \mathbf{\Sigma}^{(\ell)} \mathbf{V}^\top + \mathbf{\Delta}^{(\ell)}$ for $\ell = 1, 2, \cdots, L$, where $\mathbf{U}$ and $\mathbf{V}$ are orthogonal matrices sharing the same principal singular vector\footnotemark\ (\ie $\mathbf{u}_1 = \mathbf{v}_1$), $\mathbf{\Sigma}^{(\ell)} = \operatorname{diag}(\sigma_1^{(\ell)}, \sigma_2^{(\ell)}, \cdots)$ is the diagonal matrix of singular values satisfying $\sigma_1^{(\ell)} > \max\{1, \sigma_2^{(\ell)} + \delta\}$ and $\sigma_2^{(\ell)} \geq \eta$ for some constants $\delta, \eta > 0$, and $\mathbf{\Delta}^{(\ell)}$ is a perturbation matrix bounded by $\|\mathbf{\Delta}^{(\ell)}\|_2 \leq \epsilon \ll 1 / \sqrt{L^2 \max_{\ell} \sigma_1^{(\ell)}}$. Then, the spectral norm of the stacked attention matrices is lower-bounded by:
    \begin{equation} \label{eq:spectral-norm-attention-aggregation}
        \left\| \mathbf{S}^{(L)} \mathbf{S}^{(L - 1)} \cdots \mathbf{S}^{(1)} \right\|_2
        \geq \prod_{\ell = 1}^{L} \left( \sigma_1^{(\ell)} - \epsilon - O\left( \frac{\epsilon^2 \ell^2 \sigma_1^{(\ell)}}{\min\{\delta, \eta\}^2} \right) \right).
    \end{equation}
    For a sufficiently small $\epsilon$, the spectral norm of the attention aggregation grows exponentially with the network depth $L$.
\end{restatable}
\end{tcolorbox}
\footnotetext{These assumptions are empirically verified in \cref{tab:singular-vector-alignment,app:experimental-supplementary:supplementary-results}. The results demonstrate that the singular spaces of the attention score matrices are highly aligned across layers, and the principal left and right singular vectors ($\mathbf{u}_1$ and $\mathbf{v}_1$) exhibit a high cosine similarity of over 0.8. In fact, if each diagonal block $\mathbf{S}_u^{(\ell)}$ is a Cesàro matrix~\citep{brown1965cesaro,ross2022cesaro}, the cosine similarity between the principal singular vectors of the block-diagonal $\mathbf{S}^{(\ell)}$ will approach 1 as the length $N \to \infty$ (\cf \cref{app:theoretical-proofs:justification-of-assumption-spectral-norm-attention-aggregation}).}

\noindent
The proof is detailed in \cref{app:theoretical-proofs:proof-of-spectral-norm-attention-aggregation}. \cref{theorem:spectral-norm-attention-aggregation} reveals that if the attention scores across layers exhibit similar principal singular patterns (\eg prioritizing popular items), \emph{the spectral norm of the attention aggregation will grow exponentially as more attention layers are stacked}. This exponential growth triggers severe spectral collapse, thereby resulting in popularity bias amplification. Empirically, we observe in \cref{fig:ablation-spectral-norm}.(1) that this collapse is indeed mitigated when certain attention layers are replaced with identity mappings. This further confirms the pivotal role of attention mechanisms in driving bias amplification during model scaling.

\noindentparnoline{Low-Rank Feed-forward Projection.}
The feed-forward projection can be viewed as a product of multiple weight matrices, serving to effectively capture generalizable patterns in the representation space. However, \emph{stacking deep feed-forward layers inherently induces spectral growth}, as formalized in the following lemma:

\begin{tcolorbox}[kessokustyle=nijikayellow]
\begin{restatable}{lemma}{RestatableSpectralNormFeedForwardProjection} \label{lemma:spectral-norm-feed-forward-projection}
    \upshape
    {(\textit{Spectral Collapse in Feed-forward Projection}, Theorem 3 in~\citep{arora2019implicit}).}
    Denote the trajectory of the feed-forward weight product as $\mathbf{W}(t) = \mathbf{W}^{(1)}(t) \mathbf{W}^{(2)}(t) \cdots \mathbf{W}^{(L)}(t)$ during gradient descent training with an infinitesimally small learning rate. With a balanced initialization~\citep{arora2018optimization,arora2019implicit}, we have the following dynamics for the singular values $\sigma_k(t)$ of $\mathbf{W}(t)$:
    \begin{equation}
        \frac{\mathrm{d}}{\mathrm{d} t} \sigma_k(t) = -L \cdot (\sigma_k^2(t))^{1 - 1 / L} \cdot \langle \nabla \mathcal{L}(\mathbf{W}(t)), \mathbf{u}_k(t) \mathbf{v}_k^\top(t) \rangle,
    \end{equation}
    where $\mathbf{u}_k(t)$ and $\mathbf{v}_k(t)$ are the left and right singular vectors corresponding to $\sigma_k(t)$, and $\mathcal{L}(\cdot)$ can be any analytic loss.
\end{restatable}
\end{tcolorbox}

\noindent
\cref{lemma:spectral-norm-feed-forward-projection} reveals that the derivative $| \mathrm{d} \sigma_k(t) / \mathrm{d} t |$ of each singular value is proportional to $L \cdot (\sigma_k^2(t))^{1 - 1 / L}$, which monotonically increases \wrt the layer depth $L$. Consequently, as the model deepens (larger $L$), the spectral norm of the feed-forward projection $\|\mathbf{W}(t)\|_2 = \sigma_1(t)$ grows disproportionately faster than the remaining singular values $\sigma_{k \geq 2}(t)$, driving the representations toward an extreme low-rank state. Empirically, when we replace multiple FFNs with identity mappings, we observe a significant mitigation of spectral collapse in the final predictions, as illustrated in \cref{fig:ablation-spectral-norm}.(2). This validates that stacking more FFNs acts as the second key factor in amplifying popularity bias during model scaling.

\begin{statementbox}[title={Analysis: Spectral Collapse during Scaling}]
The \emph{popularity-prioritized attention aggregation} and \emph{low-rank feed-forward projection} synergistically exacerbate the spectral collapse of model predictions as the model deepens, which intrinsically intensifies popularity bias during model scaling.
\end{statementbox}



\section{Methodology} \label{sec:methodology}

Given the aforementioned concerns regarding popularity bias amplification during model scaling, as well as the underlying causes rooted in spectral collapse from stacking attention and feed-forward layers, blindly scaling up sequential recommenders \emph{without constraining spectral norm growth} will almost inevitably lead to sub-optimal scalability. Motivated by this critical pitfall, we propose \textbf{\mymethod} (\textbf{\uline{S}}calable \textbf{\uline{P}}opularity \textbf{\uline{R}}egularization \textbf{\uline{IN}} \textbf{\uline{T}}ransformers), a principled approach that \emph{explicitly bounds the spectral norm within key transformer components} to mitigate bias amplification and unlock sustainable scalability. In the following, we first present the algorithmic details of \mymethod in \cref{sec:methodology:details}, and subsequently provide theoretical analyses and comparative discussions in \cref{sec:methodology:analyses_and_discussions}.

\subsection{Scalable Popularity Regularization} \label{sec:methodology:details}

Recognizing spectral collapse in deep sequential recommenders as the root cause of popularity bias amplification, our approach tries to address this issue by directly constraining the spectral norm within the transformer components of each layer. Specifically, our proposed \mymethod introduces two regularization terms: \emph{attention regularization} and \emph{feed-forward regularization}.

\noindentparnoline{Attention Regularization.}
As \cref{theorem:spectral-norm-attention-aggregation} reveals, stacking multiple attention layers with similar spectral patterns (\eg prioritizing high-popularity items) inevitably drives exponential growth in the spectral norm. A natural solution to alleviate this effect is to explicitly restrict the spectral norm of the attention score matrix $\mathbf{S}^{(\ell)}$ at each layer $\ell$, thereby upper-bounding the spectral norm of the entire attention aggregation:
\begin{equation} \label{eq:attention-regularization-naive}
    \| \mathbf{S}^{(L)} \mathbf{S}^{(L - 1)} \cdots \mathbf{S}^{(1)} \|_2 \leq \prod_{\ell = 1}^{L} \| \mathbf{S}^{(\ell)} \|_2.
\end{equation}
However, directly optimizing $\| \mathbf{S}^{(\ell)} \|_2$ is highly non-trivial. Computing the exact spectral norm of this $N \times N$ matrix requires expensive SVD operations with a time complexity of $O(|\mathcal{U}| \bar{n}^3)$, which is practically infeasible for large-scale recommendation models. Furthermore, while power iteration methods~\citep{golub2013matrix,golub2000eigenvalue,yoshida2017spectral} can approximate the spectral norm of static parameter weights in quadratic time, they are inapplicable to attention scores, as these matrices are input-dependent and evolve drastically during training.

Instead, we optimize a computationally tractable surrogate upper bound for the spectral norm by leveraging the desirable row-normalized properties of attention matrices:
\begin{equation} \label{eq:attention-regularization-derivation}
    \| \mathbf{S}^{(\ell)} \|_2 \leq \sqrt{\| \mathbf{S}^{(\ell)} \|_1 \cdot \| \mathbf{S}^{(\ell)} \|_\infty} = \sqrt{\| \mathbf{S}^{(\ell)} \|_1},
\end{equation}
where $\| \cdot \|_1$ and $\| \cdot \|_\infty$ denote the matrix 1-norm (\ie maximum absolute column sum) and infinity-norm (\ie maximum absolute row sum), respectively. The last equality holds because $\mathbf{S}^{(\ell)}$ is row-normalized by the softmax function, yielding $\| \mathbf{S}^{(\ell)} \|_\infty = 1$. Note that this relaxation remains tight, given that attention scores are often concentrated on a few popular items (\cf \cref{fig:analysis-attention-scores}). Nonetheless, this surrogate bound is non-differentiable due to the \texttt{max} operation inherent in the 1-norm. To enable gradient-based optimization, we derive a smooth upper bound for the \texttt{max} operation utilizing the \texttt{logsumexp} function:
\begin{equation} \label{eq:attention-regularization}
    \| \mathbf{S}^{(\ell)} \|_1 \leq \widetilde{\| \mathbf{S}^{(\ell)} \|}_{1} \triangleq \frac{1}{\tau} \log \sum_{u, t} \exp \left( \tau \cdot \| \mathbf{S}_{u, t}^{(\ell)} \|_1 \right),
\end{equation}
where $\mathbf{S}_{u, t}^{(\ell)}$ denotes the $t$-th column vector of the attention block $\mathbf{S}_u^{(\ell)}$ for user $u$, and $\tau$ is the temperature hyperparameter. As $\tau \to +\infty$, $\widetilde{\| \mathbf{S}^{(\ell)} \|}_{1}$ converges to the exact 1-norm $\| \mathbf{S}^{(\ell)} \|_1$. In this work, we simply set $\tau = 1$ to eliminate the need for hyperparameter tuning.

Optimizing the surrogate objective in \eqnref{eq:attention-regularization} offers two compelling advantages: (i) \emph{Computational Efficiency}: Calculating the matrix 1-norm only requires summing the columns, reducing the time complexity to $O(|\mathcal{U}| \bar{n}^2)$ and circumventing time-consuming and gradient-unstable SVD operations. (ii) \emph{Intuitive Interpretability}: The column sums of the attention scores directly reflect the overall influence of each item during aggregation. By constraining the maximum column sum, we effectively prevent the model from over-relying on a few highly popular items (\cf \cref{sec:pitfalls-of-scaling-laws:analyses}). Finally, one might question why we do not use an alternative surrogate such as the Frobenius norm $\| \mathbf{S}^{(\ell)} \|_F$. Note that $\| \mathbf{S}^{(\ell)} \|_F^2 = \sum_k \sigma_k^2(\mathbf{S}^{(\ell)})$, blindly penalizing the Frobenius norm could simultaneously suppress the remaining singular values $\sigma_{k \geq 2}(\mathbf{S}^{(\ell)})$, resulting in the underutilization of critical long-tail features.

\noindentparnoline{Feed-forward Regularization.}
Another core component that contributes to spectral collapse is the feed-forward projection, where stacking more FFNs can lead to disproportionately faster growth of the spectral norm, as elaborated in \cref{lemma:spectral-norm-feed-forward-projection}. To overcome this issue, inspired by prior works on spectral regularization~\citep{yoshida2017spectral,miyato2018spectral,liu2019spectral}, we propose to penalize the spectral norm of each weight matrix $\mathbf{W}^{(\ell)}$, thereby constraining the overall spectral norm of the feed-forward projection through the following upper bound:
\begin{equation} \label{eq:feed-forward-regularization-derivation}
    \| \mathbf{W}^{(1)} \mathbf{W}^{(2)} \cdots \mathbf{W}^{(L)} \|_2 \leq \prod_{\ell = 1}^{L} \| \mathbf{W}^{(\ell)} \|_2,
\end{equation}
where $\| \mathbf{W}^{(\ell)} \|_2$ can be efficiently estimated via the power iteration~\citep{golub2013matrix,golub2000eigenvalue,yoshida2017spectral} with a time complexity of $O(d^2)$ for each layer $\ell$. Note that each feed-forward projection $\mathbf{W}^{(\ell)}$ may consist of multiple linear layers, \eg $\mathbf{W}^{(\ell)} = \mathbf{W}_1^{(\ell)} \mathbf{W}_2^{(\ell)}$. For ease of implementation, we can further relax the bound to $\| \mathbf{W}^{(\ell)} \|_2 \leq \| \mathbf{W}_1^{(\ell)} \|_2 \cdot \| \mathbf{W}_2^{(\ell)} \|_2$, opting to regularize the spectral norm of each projection matrix individually. This upper-bound relationship also holds for architectures with non-linear activations up to a constant factor, since common activation functions (\eg ReLU) typically possess a finite Lipschitz constant~\citep{yoshida2017spectral}. By applying this feed-forward regularization, we effectively curb the risk of spectral norm explosion as depth increases, thereby alleviating the spectral collapse issue.

\noindentparnoline{Overall Objective.}
Finally, combining the two aforementioned regularization terms with the specific recommendation loss $\mathcal{L}_{\text{rec}}$, the overall training objective of \textbf{\mymethod} is formulated as\footnote{We take the logarithm of the surrogate bounds derived in \eqnref{eq:attention-regularization-naive} and \eqnref{eq:feed-forward-regularization-derivation} to transform the product into a summation, facilitating layer-wise computation (\eg gradient checkpointing). For the complete formulation in more general cases (\eg with multi-head attention, residual connections, and layer normalization), please refer to \cref{app:theoretical-proofs:proof-of-spectral-norm-bound}.}:
\begin{equation} \label{eq:overall-objective}
    \mathcal{L}_{\text{\mymethod}} = \mathcal{L}_{\text{rec}} + \lambda_{\text{attn}} \sum_{\ell = 1}^{L} \log \widetilde{\| \mathbf{S}^{(\ell)} \|}_{1} + \lambda_{\text{ffn}} \sum_{\ell = 1}^{L} \log \| \mathbf{W}^{(\ell)} \|_2,
\end{equation}
where $\lambda_{\text{attn}}$ and $\lambda_{\text{ffn}}$ are hyperparameters controlling the strength of the attention and feed-forward regularizations, respectively. By unifying these regularizations with the recommendation loss, \mymethod not only improves basic recommendation accuracy but also intrinsically restrains spectral collapse \emph{within each component and across layers}. Consequently, it effectively mitigates the severe popularity bias amplification and unlocks sustainable model scalability.


\begin{table*}[t]
  \centering
    \caption{
    Overall performance comparison on SASRec++ and HSTU. The columns "N@$K$", "H@$K$", and "Fair" denote the NDCG@$K$, HitRatio@$K$, and Fair-0.8 metrics, respectively. The best results are highlighted in bold, and the strongest baselines are underlined. "\textcolor{darkred}{\textbf{Imp.\%}}" indicates the relative improvement of \mymethod over the best baseline (\cf \cref{tab:main-results-10} for @10 results).
    }
  \resizebox{\textwidth}{!}{
    \begin{tabular}{l|ccc|ccc|ccc|ccc|ccc|ccc}
    \toprule
    \multicolumn{1}{c|}{\multirow{2}[4]{*}{\textbf{Methods}}} & \multicolumn{3}{c|}{\textbf{ML-1M}} & \multicolumn{3}{c|}{\textbf{Beauty}} & \multicolumn{3}{c|}{\textbf{Toy}} & \multicolumn{3}{c|}{\textbf{Electronic}} & \multicolumn{3}{c|}{\textbf{Clothing}} & \multicolumn{3}{c}{\textbf{Book}} \\
\cmidrule(lr){2-4} \cmidrule(lr){5-7} \cmidrule(lr){8-10} \cmidrule(lr){11-13} \cmidrule(lr){14-16} \cmidrule(lr){17-19}
    & \textbf{N@5} & \textbf{H@5} & \textbf{Fair} & \textbf{N@5} & \textbf{H@5} & \textbf{Fair} & \textbf{N@5} & \textbf{H@5} & \textbf{Fair} & \textbf{N@5} & \textbf{H@5} & \textbf{Fair} & \textbf{N@5} & \textbf{H@5} & \textbf{Fair} & \textbf{N@5} & \textbf{H@5} & \textbf{Fair} \\
    \midrule
    \textbf{SASRec++} & .0807 & .1291 & .2227 & .0248 & .0406 & .2694 & .0314 & \uline{.0539} & .3741 & .0113 & .0170 & .0029 & .0018 & .0033 & .3806 & .0229 & .0443 & .4128 \\
    +DROS & .0830 & .1320 & .2172 & .0251 & .0410 & .2592 & .0282 & .0482 & .3478 & .0114 & .0175 & .0035 & .0020 & .0038 & \uline{.4157} & .0228 & .0439 & .4154 \\
    +MOJITO & .0812 & .1272 & .2188 & \uline{.0294} & .0429 & \uline{.3202} & \uline{.0363} & .0534 & .3803 & .0120 & .0194 & .0002 & .0025 & .0036 & .3717 & .0276 & .0393 & .2375 \\
    +TPAB & .0799 & .1269 & .1921 & .0220 & .0346 & .2291 & .0237 & .0363 & .3011 & .0053 & .0086 & .0004 & .0004 & .0006 & .3631 & .0252 & .0378 & .2156 \\
    +R$^2$Rec & .0771 & .1267 & .1997 & .0252 & .0426 & .3108 & .0294 & .0476 & \uline{.3878} & \uline{.0133} & \uline{.0203} & .0043 & .0019 & .0034 & .3500 & .0264 & \uline{.0471} & .3028 \\
    +D$^2$LR & \uline{.0829} & .1277 & \uline{.2317} & .0246 & .0432 & .3152 & .0288 & .0489 & .3543 & .0112 & .0169 & .0031 & .0017 & .0030 & .2994 & .0245 & .0467 & \uline{.4208} \\
    +LogDet & .0849 & \uline{.1333} & .1792 & .0272 & .0431 & .1071 & .0322 & .0507 & .2833 & .0122 & .0190 & .0006 & .0025 & .0040 & .1528 & .0249 & .0403 & .1832 \\
    +ReSN & .0739 & .1253 & .1223 & .0281 & \uline{.0452} & .1016 & .0343 & .0520 & .1515 & .0126 & .0193 & \uline{.0175} & \uline{.0032} & \uline{.0048} & .0850 & \uline{.0278} & .0379 & .3572 \\
\cmidrule(lr){1-1} \cmidrule(lr){2-4} \cmidrule(lr){5-7} \cmidrule(lr){8-10} \cmidrule(lr){11-13} \cmidrule(lr){14-16} \cmidrule(lr){17-19}
    \textbf{+\mymethod} & \textbf{.0925} & \textbf{.1405} & \textbf{.2534} & \textbf{.0355} & \textbf{.0505} & \textbf{.4127} & \textbf{.0421} & \textbf{.0584} & \textbf{.4836} & \textbf{.0161} & \textbf{.0231} & \textbf{.1618} & \textbf{.0036} & \textbf{.0055} & \textbf{.4454} & \textbf{.0347} & \textbf{.0491} & \textbf{.4953} \\
    \textcolor{darkred}{\textbf{Imp.\%}} & \textcolor{darkred}{\textbf{8.97\%}} & \textcolor{darkred}{\textbf{5.35\%}} & \textcolor{darkred}{\textbf{2.17\%}} & \textcolor{darkred}{\textbf{21.09\%}} & \textcolor{darkred}{\textbf{11.83\%}} & \textcolor{darkred}{\textbf{9.25\%}} & \textcolor{darkred}{\textbf{15.81\%}} & \textcolor{darkred}{\textbf{8.25\%}} & \textcolor{darkred}{\textbf{9.58\%}} & \textcolor{darkred}{\textbf{20.52\%}} & \textcolor{darkred}{\textbf{14.08\%}} & \textcolor{darkred}{\textbf{14.43\%}} & \textcolor{darkred}{\textbf{13.90\%}} & \textcolor{darkred}{\textbf{15.05\%}} & \textcolor{darkred}{\textbf{2.97\%}} & \textcolor{darkred}{\textbf{24.67\%}} & \textcolor{darkred}{\textbf{4.37\%}} & \textcolor{darkred}{\textbf{7.45\%}} \\
    \midrule
    \textbf{HSTU} & .0784 & .1250 & .1902 & .0267 & .0417 & .3369 & .0310 & .0473 & .4215 & .0159 & .0248 & .0014 & .0022 & .0035 & .4179 & .0269 & .0393 & .1815 \\
    +DROS & .0802 & .1267 & .1934 & .0272 & .0437 & \uline{.3861} & \uline{.0320} & .0473 & \uline{.4951} & .0160 & .0248 & .0009 & .0023 & .0039 & \uline{.4582} & .0276 & .0398 & .1739 \\
    +MOJITO & .0799 & .1299 & .2035 & .0277 & .0422 & .3174 & .0318 & .0441 & .3753 & .0168 & .0264 & .0001 & .0026 & .0039 & .3856 & .0282 & .0412 & .1768 \\
    +TPAB & .0751 & .1219 & .2240 & \uline{.0286} & \uline{.0443} & .3781 & .0299 & .0448 & .4133 & .0114 & .0188 & .0001 & .0010 & .0015 & .3834 & .0290 & .0429 & .1383 \\
    +R$^2$Rec & .0705 & .1180 & .2128 & .0255 & .0423 & .3269 & .0312 & .0475 & .4429 & .0169 & .0263 & .0021 & .0024 & .0038 & .3430 & .0272 & .0417 & .1679 \\
    +D$^2$LR & .0798 & .1279 & \uline{.2305} & .0265 & .0428 & .3824 & .0314 & \uline{.0497} & .4589 & .0164 & .0267 & .0084 & .0023 & .0039 & .4508 & \uline{.0294} & .0426 & \uline{.2490} \\
    +LogDet & \uline{.0828} & \uline{.1306} & .1574 & .0282 & .0430 & .3110 & .0319 & .0496 & .4425 & .0177 & \uline{.0273} & \uline{.0116} & \uline{.0028} & .0041 & .3155 & \uline{.0294} & \uline{.0435} & .1029 \\
    +ReSN & .0740 & .1223 & .1520 & .0256 & .0418 & .1398 & .0289 & .0439 & .2350 & \uline{.0186} & .0257 & .0000 & .0027 & \uline{.0043} & .1042 & .0229 & .0295 & .0752 \\
\cmidrule(lr){1-1} \cmidrule(lr){2-4} \cmidrule(lr){5-7} \cmidrule(lr){8-10} \cmidrule(lr){11-13} \cmidrule(lr){14-16} \cmidrule(lr){17-19}
    \textbf{+\mymethod} & \textbf{.0871} & \textbf{.1357} & \textbf{.2503} & \textbf{.0362} & \textbf{.0514} & \textbf{.4449} & \textbf{.0438} & \textbf{.0589} & \textbf{.5173} & \textbf{.0204} & \textbf{.0297} & \textbf{.2150} & \textbf{.0036} & \textbf{.0055} & \textbf{.4744} & \textbf{.0346} & \textbf{.0481} & \textbf{.3237} \\
    \textcolor{darkred}{\textbf{Imp.\%}} & \textcolor{darkred}{\textbf{5.14\%}} & \textcolor{darkred}{\textbf{3.90\%}} & \textcolor{darkred}{\textbf{1.99\%}} & \textcolor{darkred}{\textbf{26.79\%}} & \textcolor{darkred}{\textbf{16.07\%}} & \textcolor{darkred}{\textbf{5.89\%}} & \textcolor{darkred}{\textbf{36.98\%}} & \textcolor{darkred}{\textbf{18.45\%}} & \textcolor{darkred}{\textbf{2.22\%}} & \textcolor{darkred}{\textbf{9.72\%}} & \textcolor{darkred}{\textbf{8.75\%}} & \textcolor{darkred}{\textbf{20.34\%}} & \textcolor{darkred}{\textbf{30.41\%}} & \textcolor{darkred}{\textbf{28.31\%}} & \textcolor{darkred}{\textbf{1.62\%}} & \textcolor{darkred}{\textbf{17.75\%}} & \textcolor{darkred}{\textbf{10.55\%}} & \textcolor{darkred}{\textbf{7.47\%}} \\
    \bottomrule
    \end{tabular}
    }
  \label{tab:main-results-5}
\end{table*}

\subsection{Analyses and Discussions} \label{sec:methodology:analyses_and_discussions}

\noindentparnoline{Theoretical Guarantees.}
To theoretically verify the effectiveness of our method, we establish the following upper-bound relationship between the spectral norm of model predictions and our proposed \mymethod objective \eqnref{eq:overall-objective} to provide a rigorous guarantee for spectral norm control and popularity bias mitigation:

\begin{tcolorbox}[kessokustyle=bocchipink]
\begin{restatable}{theorem}{RestateTheoremSpectralNormBound} \label{theorem:spectral-norm-bound}
    \upshape
    {(\textit{Theoretical Guarantee of \mymethod}).}
    With proper hyperparameters $\lambda_{\text{attn}}$ and $\lambda_{\text{ffn}}$, the \mymethod regularization terms establish a rigorous upper bound on the spectral norm $\| \hat{\mathbf{Y}} \|_2$ of the final model predictions.
\end{restatable}
\end{tcolorbox}

\noindent
The proof is provided in \cref{app:theoretical-proofs:proof-of-spectral-norm-bound}. \cref{theorem:spectral-norm-bound} ensures that minimizing the regularizers in \mymethod effectively constrains the spectral norm of model predictions, thereby providing a theoretical guarantee for mitigating popularity bias (\cf \cref{lemma:popularity-bias-spectral-norm}).

\noindentparnoline{Computational Efficiency.}
\mymethod is highly computationally efficient. Specifically, the attention regularization merely involves calculating the column sums of the already-computed attention scores, while the feed-forward regularization is efficiently estimated via the power iteration method\footnote{Following prior works~\citep{yoshida2017spectral,miyato2018spectral}, we do only one power iteration step per forward pass.} at each layer. Consequently, the additional time complexity incurred by \mymethod is $O(L|\mathcal{U}|\bar{n}^2 + L d^2)$, which is negligible compared to the overall forward and backward passes. Empirically, as illustrated in \cref{fig:runtime-comparison}, we observe that \mymethod incurs only about 3\% additional training time. This overhead is highly acceptable given the substantial improvements in both recommendation accuracy and fairness.

\noindentparnoline{Discussions on CF Methods.}
Existing methods in traditional collaborative filtering (CF) have explored mitigating popularity bias by regularizing the spectral norm, which conceptually aligns with our motivation. For instance, ReSN~\citep{lin2025recommendation} proposes to approximate and penalize the spectral norm by factorizing it based on user and item embeddings; LogDet~\citep{zhang2023mitigating} applies redundancy decorrelation regularization to user/item covariance matrices to mitigate embedding collapse in graph-based CF. While these methods achieve promising results in CF scenarios, directly adapting them to sequential recommendation presents two fundamental limitations: (i) They typically rely on static user embeddings. However, in sequential recommendation, user preferences are dynamic and can evolve drastically over time, which renders the optimization of global spectral penalties highly unstable. (ii) They treat the model as a black box and fail to address the unique mechanisms of the transformer architecture, where spectral collapse is inherently exacerbated by the deep stacking of attention and feed-forward layers. In contrast, \mymethod adopts a targeted, layer-wise regularization strategy that explicitly bounds the spectral norm within the specific components responsible for the collapse, thereby enabling stable and highly effective optimization for deep sequential recommenders.



\begin{figure*}[t]
    \centering
    \includegraphics[width=0.95\textwidth]{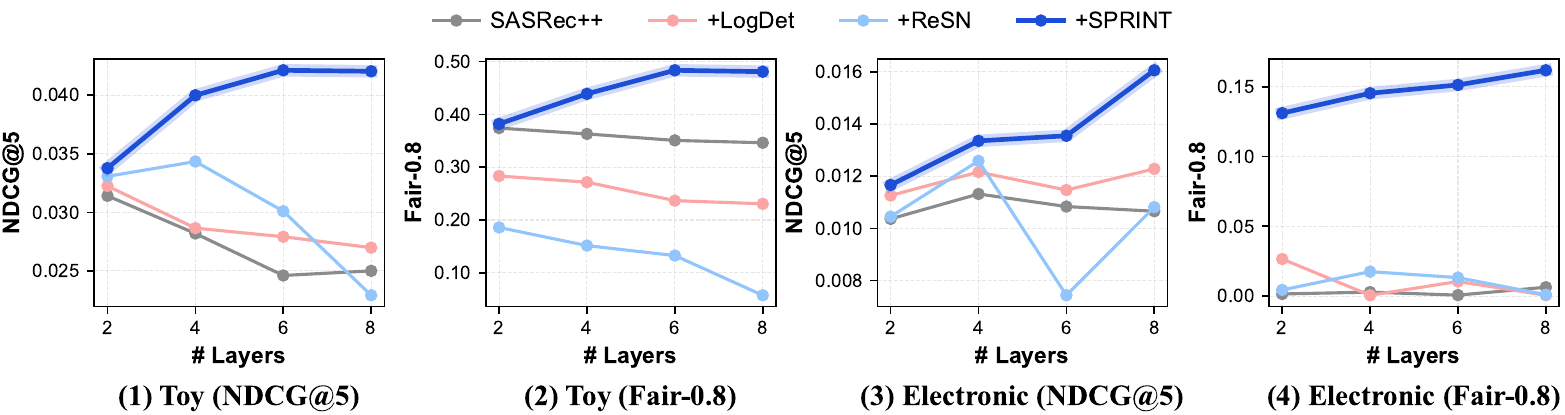}
    \caption{
        Layer scaling results of \mymethod and representative baselines on the SASRec++ backbone, where the number of layers is varied from 2 to 8, and the accuracy and fairness metrics are reported. Refer to \cref{fig:layer-scaling-ml1m-beauty} for results on other datasets.
    }
    \Description{Layer scaling results.}
    \label{fig:layer-scaling-toy-electronic}
\end{figure*}

\begin{figure*}[t]
    \centering
    \includegraphics[width=\textwidth]{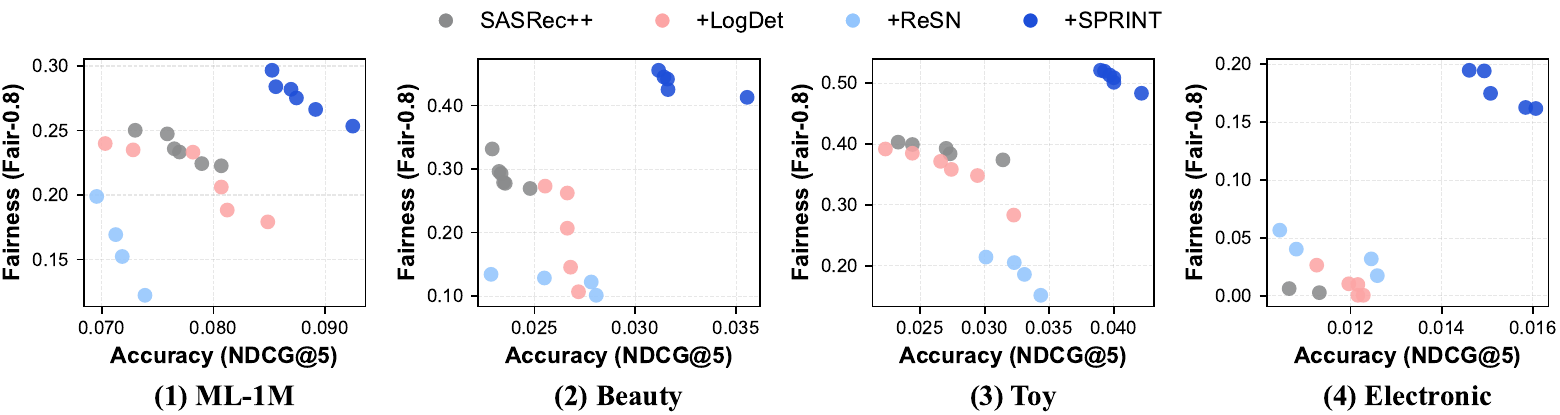}
    \caption{
        Pareto frontier of accuracy and fairness for \mymethod and representative baselines on the SASRec++ backbone.
    }
    \Description{Pareto frontier of accuracy and fairness.}
    \label{fig:pareto}
\end{figure*}

\section{Experiments} \label{sec:experiments}

In this section, we conduct extensive experiments to answer the following research questions (RQs):

\begin{itemize}[topsep=3pt,leftmargin=10pt,itemsep=0pt]
    \item \textbf{RQ1}: How does \mymethod perform compared to existing methods in terms of recommendation accuracy and fairness?
    \item \textbf{RQ2}: How does \mymethod perform under different model scales, especially when stacking more transformer layers?
    \item \textbf{RQ3}: How do the attention and feed-forward regularizations contribute to the overall performance of \mymethod?
    \item \textbf{RQ4}: Can \mymethod generalize to other recommendation scenarios, \eg generative recommendation?
\end{itemize}


\subsection{Experimental Setup} \label{sec:experiments:settings}

\noindentparnolinenospace{Datasets.}
To ensure a fair comparison, we conduct experiments on six widely used datasets from MovieLens~\citep{harper2015movielens} and Amazon~\citep{he2016ups,mcauley2015image}, specifically including ML-1M/20M, Beauty, Toy, Electronic, Clothing, and Book. The dataset preprocessing strictly follows the previous studies~\citep{kang2018self,yang2024psl}; further details are provided in \cref{app:experimental-supplementary:datasets}.

\noindentparnoline{Recommendation Backbones.}
We evaluate each method on two advanced transformer-based sequential recommendation backbones, namely SASRec++~\citep{zhang2024scaling,guo2024scaling} (2024) and HSTU~\citep{zhai2024actions} (2024). In the main experiments, we scale the number of layers up to 8 and the embedding dimensionality up to 256. These scale settings deliberately exceed the typical configurations used in prior studies, allowing us to better assess the scalability of different methods. Furthermore, we conduct large-scale experiments on the ML-20M dataset using SASRec++ with up to 0.34B non-embedding parameters (12 layers and 1536-dimensional embeddings), as illustrated in \cref{fig:ml-20m-scaling-laws}. Please refer to \cref{app:experimental-supplementary:backbone-details} for complete backbone details.

\noindentparnoline{Compared Methods.} To validate the effectiveness of \mymethod, we compare it against a diverse set of representative debiasing methods from the following categories: (i) \textit{spectral regularization}, \eg LogDet~\citep{zhang2023mitigating} (NeurIPS '23) and ReSN~\citep{lin2025recommendation} (WSDM '25); (ii) \textit{propensity reweighting}, \eg R$^2$Rec~\citep{li2025reembedding} (WWW '25) and D$^2$LR~\citep{lu2025dual} (SIGIR '25); (iii) \textit{temporal-aware disentanglement}, \eg MOJITO~\citep{tran2023attention} (SIGIR '23) and TPAB~\citep{yoo2025generalizable} (KDD '25); and (iv) \textit{distributionally robust optimization}, \eg DROS~\citep{yang2023generic} (SIGIR '23). We closely follow the original implementations and settings of these methods, adapting them to the sequential recommendation scenario when necessary. Detailed implementations and hyperparameter configurations are provided in \cref{app:experimental-supplementary:compared-methods}.

\noindentparnoline{Evaluation Metrics.}
We evaluate recommendation performance from the perspectives of both accuracy and fairness. For accuracy, we report the widely used NDCG@$K$ and HitRatio@$K$ metrics, where $K \in \{5, 10\}$. For fairness, we mainly adopt the Fairness metric Fair-$\alpha$~\citep{lin2025recommendation} measuring the exposure ratio of bottom-$\alpha$ long-tail items ($\alpha = 0.8$ in the main experiments, \cf \cref{sec:preliminaries:popularity-bias}). Refer to \cref{tab:fairness-results} for additional fairness metrics and results.



\subsection{Experimental Results} \label{sec:experiments:results}

\noindentparnolinenospace{(RQ1) Overall Performance.}
\cref{tab:main-results-5,tab:main-results-10} present the overall recommendation performance comparison between \mymethod and the baseline methods on the SASRec++ and HSTU backbones. As revealed, \mymethod consistently outperforms all compared methods across all six datasets, yielding an average improvement of \emph{+15.70\%} in NDCG/HR@5 and \emph{+7.12\%} in Fair-0.8. Furthermore, as illustrated in \cref{fig:pareto}, \mymethod establishes a strictly \emph{superior Pareto frontier} in the accuracy-fairness trade-off compared to all baselines. This clearly demonstrates the superiority of \mymethod in simultaneously enhancing recommendation quality and mitigating popularity bias. \cref{tab:fairness-results} reports results on additional fairness metrics.

\noindentparnoline{(RQ2) Layer Scaling Performance.}
As theoretically revealed in \cref{sec:pitfalls-of-scaling-laws:analyses}, the popularity bias amplification induced by spectral collapse becomes increasingly severe as more layers are stacked, fundamentally bottlenecking scalability. To examine the effectiveness of \mymethod within deep transformer architectures, we conduct layer-scaling experiments, the results of which are depicted in \cref{fig:layer-scaling-toy-electronic} and \cref{fig:layer-scaling-ml1m-beauty}. The results demonstrate that our method consistently achieves better accuracy and improved fairness as the model deepens. In contrast, the performance of baseline methods typically stagnates or even degrades due to the unconstrained amplification of popularity bias. This empirical observation perfectly coincides with our discussion in \cref{sec:methodology:analyses_and_discussions}: the targeted, layer-wise regularization design of \mymethod effectively curbs spectral collapse in deep networks, thereby unlocking sustainable scalability.

\noindentparnoline{(RQ3) Ablation Study and Hyperparameter Sensitivity.}
As shown in \cref{tab:ablation-study}, both the attention and feed-forward regularizations contribute substantially to the overall performance of \mymethod, with the optimal results achieved when both regularizations are jointly applied. We further analyze the sensitivity of \mymethod to its key hyperparameters, $\lambda_{\text{attn}}$ and $\lambda_{\text{ffn}}$, which control the strength of the respective regularizations. Results in \cref{fig:hyper-parameter-sensitivity} exhibit a clear performance peak, highlighting the dynamic trade-off between accuracy and fairness. Notably, \mymethod remains highly robust across a reasonably wide range of hyperparameter values, facilitating its practical deployment without the need for exhaustive fine-tuning.

\noindentparnoline{(RQ4) Generalization to Generative Recommendation.}
Finally, we extend \mymethod to the emerging paradigm of generative recommendation, where each item is represented by a unique sequence of semantic IDs. Specifically, we adopt the widely recognized TIGER~\citep{rajput2023recommender} and LETTER~\citep{wang2024learnable} models as our backbones, and seamlessly apply \mymethod to regularize the self-attention and feed-forward blocks within their underlying T5~\citep{raffel2020exploring} architectures. As reported in \cref{tab:generative-recommendation-results}, \mymethod consistently yields notable performance gains across both generative backbones. This highlights the strong generalization capability of our proposed method.

\begin{table}[t]
  \centering
  \caption{\mymethod in generative recommendation.}
  \resizebox{\columnwidth}{!}{
    \begin{tabular}{l|ccc|ccc}
    \toprule
    \multicolumn{1}{c|}{\multirow{2}[4]{*}{\textbf{Backbones}}} & \multicolumn{3}{c|}{\textbf{Beauty}} & \multicolumn{3}{c}{\textbf{Toy}} \\
\cmidrule(lr){2-4} \cmidrule(lr){5-7}
    & \textbf{NDCG@5} & \textbf{HR@5} & \textbf{Fair} & \textbf{NDCG@5} & \textbf{HR@5} & \textbf{Fair} \\
    \midrule
    \textbf{TIGER} & 0.0306 & 0.0450 & 0.1081 & 0.0318 & 0.0487 & 0.2055 \\
    \textbf{+\mymethod} & \textbf{0.0343} & \textbf{0.0503} & \textbf{0.1609} & \textbf{0.0328} & \textbf{0.0504} & \textbf{0.2597} \\
    \midrule
    \textbf{LETTER} & 0.0301 & 0.0450 & 0.0596 & 0.0330 & 0.0473 & 0.2162 \\
    \textbf{+\mymethod} & \textbf{0.0324} & \textbf{0.0482} & \textbf{0.1216} & \textbf{0.0341} & \textbf{0.0493} & \textbf{0.2671} \\
    \bottomrule
    \end{tabular}
  }
  \label{tab:generative-recommendation-results}
\end{table}



\section{Related Work} \label{sec:related_work}

\noindentparnolinenospace{Sequential Recommendation.}
Sequential recommendation~\citep{chang2021sequential,chen2018sequential}, which aims to model dynamic user preferences and temporal dependencies in user-item interactions, has long been a fundamental research topic in recommender systems. Early works in this area primarily focus on leveraging RNNs~\citep{hidasi2015session} or CNNs~\citep{tang2018personalized} to capture sequential patterns. With the advent of the attention mechanism~\citep{vaswani2017attention}, transformer-based models~\citep{kang2018self,sun2019bert4rec} have become the dominant paradigm in sequential recommendation due to their superior capability in modeling long-range dependencies and facilitating parallel computation. In recent years, by aggressively scaling model sizes to billions of parameters~\citep{zhai2024actions,zivic2024scaling,shen2024optimizing,zhang2024scaling,chai2025longer,xu2025climber}, the capabilities of sequential recommenders have been significantly advanced, achieving remarkable improvements in recommendation quality. However, while these scaling efforts yield substantial performance gains, they inherently exacerbate popularity bias. This severe side effect has been largely overlooked in the literature, and our work precisely aims to fill this critical gap.

\noindentparnoline{Popularity Bias and Debiasing.}
In recommender systems, popularity bias refers to the notorious phenomenon where popular items are recommended disproportionately more frequently than niche ones, inevitably leading to unfairness and reduced diversity~\citep{chen2023bias}. Several lines of research have explored the causes of this bias, attributing it to confounding causal factors~\citep{wang2021deconfounded,wei2021model}, long-tail data distributions~\citep{zhang2021causal}, and low-rank embedding spaces~\citep{ohsaka2023curse,zhang2023mitigating}. Given its detrimental impacts, various debiasing strategies have been proposed. One prominent approach utilizes causal inference techniques to disentangle popularity bias from genuine user preferences~\citep{ding2022addressing,li2023balancing,yoo2025generalizable}. Yet, these methods typically require carefully crafted causal graphs, which are difficult to generalize across diverse scenarios. Another major direction involves propensity-based reweighting~\citep{gruson2019offline,zhang2023model,lu2025dual}, which estimates the propensity scores of items to reweight training samples. However, blindly penalizing popular signals can easily result in suboptimal performance. Other efforts address popularity distribution shifts via temporal-aware modeling~\citep{tran2023attention,yang2023generic}; unfortunately, they merely capture the temporal dynamics of the bias without tackling its amplification. Recently, regularization-based methods focusing on decorrelation and spectral collapse~\citep{zhu2021popularity,zhang2023mitigating,lin2025recommendation} have been proposed to constrain the optimization trajectory toward more balanced solutions, achieving promising results in standard collaborative filtering. Nonetheless, directly adapting these techniques to sequential recommendation is highly non-trivial due to the unique bias amplification mechanisms inherent in deep transformer architectures. Our work bridges this gap by rigorously analyzing the underlying architectural causes of this amplification and introducing a layer-wise regularization framework to effectively mitigate the issue, thereby unlocking sustainable scalability for deep sequential recommenders.


\section{Conclusion and Future Directions} \label{sec:conclusion_and_future_directions}

This work identifies a critical pitfall in scaling transformer-based sequential recommenders: although increasing model size generally improves recommendation accuracy, it simultaneously amplifies popularity bias. Our theoretical and empirical analyses reveal that this phenomenon is fundamentally rooted in the spectral collapse induced by stacking more attention and feed-forward layers within the transformer. To address this issue, we propose \mymethod, a simple yet effective regularization framework that constrains the maximum column sums of the attention matrix alongside the spectral norms of the feed-forward parameters. Theoretical analysis guarantees that \mymethod effectively curbs popularity bias by strictly upper-bounding the overall spectral norm of model predictions. Extensive experiments demonstrate that \mymethod consistently outperforms various representative baselines in terms of both accuracy and fairness. Crucially, \mymethod yields significantly more favorable scaling behaviors, achieving superior accuracy scaling while maintaining a balanced popularity distribution.

Looking forward, exploring more computationally efficient implementations of \mymethod (\eg integrating it with state-of-the-art efficient attention mechanisms like FlashAttention~\citep{dao2024flashattention}) presents a highly promising direction. Furthermore, it would be intriguing to investigate how other scaling dimensions (\eg the number of attention heads, embedding dimensionality, and training data volume) intricately impact the recommendation fairness.



\begin{acks}
This work is supported by the National Natural Science Foundation of China (62372399, 62476244), the Starry Night Science Fund of Zhejiang University Shanghai Institute for Advanced Study (SN-ZJU-SIAS-001), and the advanced computing resources provided by the Supercomputing Center of Hangzhou City University.
\end{acks}



\clearpage
\bibliographystyle{ACM-Reference-Format}
\bibliography{references}




\appendix


\section{Theoretical Proofs} \label{app:theoretical-proofs}


\subsection{Proof of \cref{theorem:spectral-norm-attention-aggregation}} \label{app:theoretical-proofs:proof-of-spectral-norm-attention-aggregation}

To prove \cref{theorem:spectral-norm-attention-aggregation}, we first present two lemmas establishing the stability of the spectral norm and singular vectors under perturbations.

\begin{tcolorbox}[kessokustyle=nijikayellow]
\begin{restatable}{lemma}{RestatableApproximateSVDPerturbation} \label{lemma:approximate-svd-perturbation}
    \upshape
    {(\textit{Stability of Spectral Norm under Perturbation}).}
    Consider a matrix $\mathbf{S}$ and its perturbed SVD $\mathbf{S} = \mathbf{U} \mathbf{\Sigma} \mathbf{V}^\top + \mathbf{\Delta}$, where $\mathbf{U}, \mathbf{V}$ are orthogonal matrices, $\mathbf{\Sigma} = \operatorname{diag}(\sigma_1, \sigma_2, \cdots)$ is the diagonal matrix of singular values, and $\mathbf{\Delta}$ is the perturbation satisfying $\|\mathbf{\Delta}\|_2 \leq \epsilon$. Then, the spectral norm of $\mathbf{S}$ is lower-bounded by:
    \begin{equation} \label{eq:approximate-svd-perturbation-spectral-norm-lower-bound}
        \|\mathbf{S}\|_2 \geq \sigma_1 - \epsilon.
    \end{equation}
\end{restatable}
\end{tcolorbox}

\begin{proof}[Proof of \cref{lemma:approximate-svd-perturbation}]
Let $\mathbf{v}_1$ be the principal right singular vector of the unperturbed component $\mathbf{U} \mathbf{\Sigma} \mathbf{V}^\top$, which corresponds exactly to the first column of $\mathbf{V}$. Then, we have:
\begin{equation}
    \mathbf{S} \mathbf{v}_1 = \mathbf{U} \mathbf{\Sigma} \mathbf{V}^\top \mathbf{v}_1 + \mathbf{\Delta} \mathbf{v}_1 = \sigma_1 \mathbf{u}_1 + \mathbf{\Delta} \mathbf{v}_1.
\end{equation}
By the triangle inequality, it follows that:
\begin{equation}
    \|\mathbf{S} \mathbf{v}_1\|_2 \geq \sigma_1 - \|\mathbf{\Delta} \mathbf{v}_1\|_2 \geq \sigma_1 - \epsilon.
\end{equation}
By the definition of the spectral norm, we obtain:
\begin{equation}
    \|\mathbf{S}\|_2 = \max_{\|\mathbf{v}\|_2 = 1} \|\mathbf{S} \mathbf{v}\|_2 \geq \|\mathbf{S} \mathbf{v}_1\|_2 \geq \sigma_1 - \epsilon,
\end{equation}
which completes the proof.
\end{proof}

\begin{tcolorbox}[kessokustyle=nijikayellow]
\begin{restatable}{lemma}{RestatableApproximateSVDPerturbationSin} \label{lemma:approximate-svd-perturbation-sin}
    \upshape
    {(\textit{Stability of Singular Vectors under Perturbation}).}
    Under the same notation and assumptions as in \cref{lemma:approximate-svd-perturbation}, let $(\hat{\mathbf{u}}_1, \hat{\mathbf{v}}_1)$ be the principal singular vectors of $\mathbf{S}$, and $(\mathbf{u}_1, \mathbf{v}_1)$ be the principal singular vectors of the unperturbed component $\mathbf{U} \mathbf{\Sigma} \mathbf{V}^\top$. Assume that the spectral gap between the first and second largest singular values of $\mathbf{U} \mathbf{\Sigma} \mathbf{V}^\top$ satisfies $\sigma_1 \geq \sigma_2 + \delta$ for some constant $\delta > 0$. Then, there exists a constant $C = 2 + \epsilon / \sigma_1$ such that:
    \begin{equation} \label{eq:approximate-svd-perturbation-sin}
        \sin \theta(\hat{\mathbf{u}}_1, \mathbf{u}_1) \leq \frac{C \epsilon}{\delta}, \quad
        \sin \theta(\hat{\mathbf{v}}_1, \mathbf{v}_1) \leq \frac{C \epsilon}{\delta},
    \end{equation}
    where $\theta(\cdot, \cdot)$ denotes the angle between two vectors.
\end{restatable}
\end{tcolorbox}

\begin{proof}[Proof of \cref{lemma:approximate-svd-perturbation-sin}]
We denote the unperturbed component as $\mathbf{S}^* = \mathbf{U} \mathbf{\Sigma} \mathbf{V}^\top$ for brevity. Then,
\begin{equation}
    \mathbf{S}^\top \mathbf{S} - (\mathbf{S}^*)^\top \mathbf{S}^* = (\mathbf{S}^*)^\top \mathbf{\Delta} + \mathbf{\Delta}^\top \mathbf{S}^* + \mathbf{\Delta}^\top \mathbf{\Delta}.
\end{equation}
Observe that $\| (\mathbf{S}^*)^\top \mathbf{\Delta} \|_2 \leq \sigma_1 \epsilon$, $\| \mathbf{\Delta}^\top \mathbf{S}^* \|_2 \leq \sigma_1 \epsilon$, and $\| \mathbf{\Delta}^\top \mathbf{\Delta} \|_2 \leq \epsilon^2$. Therefore, we have:
\begin{equation}
    \| \mathbf{S}^\top \mathbf{S} - (\mathbf{S}^*)^\top \mathbf{S}^* \|_2 \leq 2 \sigma_1 \epsilon + \epsilon^2.
\end{equation}
Recalling the assumption $\sigma_1 \geq \sigma_2 + \delta$, and noting that the eigenvalues of $(\mathbf{S}^*)^\top \mathbf{S}^*$ are $\lambda_k = \sigma_k^2$, we can lower-bound the spectral gap by:
\begin{equation}
    \lambda_1 - \lambda_2 = (\sigma_1 - \sigma_2)(\sigma_1 + \sigma_2) \geq \delta \cdot \sigma_1.
\end{equation}
Applying the Davis-Kahan $\sin \theta$ Theorem (\cf Theorem 1 in~\citep{yu2015useful}) to $\mathbf{S}^\top \mathbf{S}$ and $(\mathbf{S}^*)^\top \mathbf{S}^*$, we obtain:
\begin{equation}
    \sin \theta(\hat{\mathbf{v}}_1, \mathbf{v}_1) \leq \frac{\| \mathbf{S}^\top \mathbf{S} - (\mathbf{S}^*)^\top \mathbf{S}^* \|_2}{\lambda_1 - \lambda_2} \leq \frac{2 \sigma_1 \epsilon + \epsilon^2}{\delta \sigma_1} = \frac{C \epsilon}{\delta},
\end{equation}
where $C = 2 + \epsilon / \sigma_1$. This proves the bound for the right singular vectors. Similarly, we can derive the same bound for the left singular vectors $\hat{\mathbf{u}}_1$ and $\mathbf{u}_1$ by applying the Davis-Kahan $\sin \theta$ Theorem to $\mathbf{S} \mathbf{S}^\top$ and $\mathbf{S}^* (\mathbf{S}^*)^\top$. This completes the proof.
\end{proof}

Now we are ready to prove \cref{theorem:spectral-norm-attention-aggregation}.

\begin{tcolorbox}[kessokustyle=bocchipink]
\RestatableSpectralNormAttentionAggregation*
\end{tcolorbox}
\footnotetext{These assumptions are empirically verified in \cref{tab:singular-vector-alignment,app:experimental-supplementary:supplementary-results}. The results demonstrate that the singular spaces of the attention score matrices are highly aligned across layers, and the principal left and right singular vectors ($\mathbf{u}_1$ and $\mathbf{v}_1$) exhibit a high cosine similarity of over 0.8. In fact, if each diagonal block $\mathbf{S}_u^{(\ell)}$ is a Cesàro matrix~\citep{brown1965cesaro,ross2022cesaro}, the cosine similarity between the principal singular vectors of the block-diagonal $\mathbf{S}^{(\ell)}$ will approach 1 as the length $N \to \infty$ (\cf \cref{app:theoretical-proofs:justification-of-assumption-spectral-norm-attention-aggregation}).}

\begin{proof}[Proof of \cref{theorem:spectral-norm-attention-aggregation}]
The core idea of this proof is to establish a lower bound for the spectral norm $\left\| \mathbf{S}^{(L)} \mathbf{S}^{(L - 1)} \cdots \mathbf{S}^{(1)} \right\|_2$ by evaluating the \emph{projection norm} of the attention aggregation along the principal singular vector $\mathbf{v}_1$ of the unperturbed component $\mathbf{U} \mathbf{\Sigma}^{(\ell)} \mathbf{V}^\top$ for each layer $\ell$.

\noindentparnoline{Projection Norm Decomposition.}
We first define a sequence of unit vectors $\mathbf{v}^{(\ell)}$ for each layer $\ell = 1, 2, \cdots, L$ as follows:
\begin{equation} \label{eq:recursive-vector-definition}
    \mathbf{v}^{(\ell)} = \frac{\mathbf{S}^{(\ell)} \mathbf{v}^{(\ell - 1)}}{\|\mathbf{S}^{(\ell)} \mathbf{v}^{(\ell - 1)}\|_2}, \quad \mathbf{v}^{(0)} = \mathbf{v}_1,
\end{equation}
where $\mathbf{v}_1 = \mathbf{u}_1$ is the principal singular vector of the unperturbed component $\mathbf{U} \mathbf{\Sigma}^{(\ell)} \mathbf{V}^\top$ (note that by our assumption, $\mathbf{v}_1$ and $\mathbf{u}_1$ are highly aligned and are shared across layers in practice). This recursive definition yields a concise decomposition of the projection norm along the principal singular vector $\mathbf{v}^{(0)} = \mathbf{v}_1$:
\begin{equation} \label{eq:projection-norm-decomposition}
    \left\| \mathbf{S}^{(L)} \mathbf{S}^{(L - 1)} \cdots \mathbf{S}^{(1)} \mathbf{v}^{(0)} \right\|_2
    = \prod_{\ell = 1}^{L} \left\| \mathbf{S}^{(\ell)} \mathbf{v}^{(\ell - 1)} \right\|_2.
\end{equation}
To estimate the \emph{intermediate projection norm} $\left\| \mathbf{S}^{(\ell)} \mathbf{v}^{(\ell - 1)} \right\|_2$, we let $\theta^{(\ell)} = \theta(\mathbf{v}^{(\ell)}, \mathbf{v}_1)$ denote the angle between the direction $\mathbf{v}^{(\ell)}$ at layer $\ell$ and the unperturbed principal singular vector $\mathbf{v}_1$. We can then decompose the direction $\mathbf{v}^{(\ell - 1)}$ into a component aligned with $\mathbf{v}_1$ and an orthogonal component $\mathbf{w}^{(\ell - 1)}$ as follows:
\begin{equation} \label{eq:vector-orthogonal-decomposition}
    \mathbf{v}^{(\ell - 1)} = \cos \theta^{(\ell - 1)} \mathbf{v}_1 + \sin \theta^{(\ell - 1)} \mathbf{w}^{(\ell - 1)},
\end{equation}
where $\mathbf{w}^{(\ell - 1)} \perp \mathbf{v}_1$ and $\|\mathbf{w}^{(\ell - 1)}\|_2 = 1$. Recalling that $\mathbf{u}_1 = \mathbf{v}_1$, we can expand the intermediate projection as follows:
\begin{equation}
\begin{aligned}
    \mathbf{S}^{(\ell)} \mathbf{v}^{(\ell - 1)} & = \sigma_1^{(\ell)} \cos \theta^{(\ell - 1)} \mathbf{v}_1 \\
    & + \sin \theta^{(\ell - 1)} \mathbf{U} \mathbf{\Sigma}^{(\ell)} \mathbf{V}^\top \mathbf{w}^{(\ell - 1)} + \mathbf{\Delta}^{(\ell)} \mathbf{v}^{(\ell - 1)}.
\end{aligned}
\end{equation}
Since the first two terms on the right-hand side are orthogonal, we can invoke the reverse triangle inequality to lower-bound the intermediate projection norm:
\begin{equation} \label{eq:projection-norm-decomposition-lower-bound}
\begin{aligned}
    \left\| \mathbf{S}^{(\ell)} \mathbf{v}^{(\ell - 1)} \right\|_2
    & \geq \sigma_1^{(\ell)} \cos \theta^{(\ell - 1)} - \left\| \mathbf{\Delta}^{(\ell)} \mathbf{v}^{(\ell - 1)} \right\|_2 \\
    & \geq \sigma_1^{(\ell)} \cos \theta^{(\ell - 1)} - \epsilon,
\end{aligned}
\end{equation}
where the last inequality follows from $\|\mathbf{\Delta}^{(\ell)}\|_2 \leq \epsilon$, as specified in our perturbation assumptions.

\noindentparnoline{Bounding the Unperturbed Angle.}
Inequality \eqnref{eq:projection-norm-decomposition-lower-bound} establishes a concise lower bound for the intermediate projection norm at each layer $\ell$, wherein the term $\cos \theta^{(\ell - 1)}$ requires further evaluation. To this end, we proceed to bound the angle $\theta^{(\ell)}$ based on the angle $\theta^{(\ell - 1)}$ from the preceding layer. Specifically, we first consider the ideal case where there is no perturbation in the $\ell$-th layer, \ie $\mathbf{\Delta}^{(\ell)} = \mathbf{0}$. In this scenario, our goal is to bound the corresponding ideal angle $\theta_{\text{ideal}}^{(\ell)} = \theta(\mathbf{U} \mathbf{\Sigma}^{(\ell)} \mathbf{V}^\top \mathbf{v}^{(\ell - 1)}, \mathbf{v}_1)$.

To achieve this, we decompose the orthogonal component $\mathbf{w}^{(\ell - 1)}$ from \eqnref{eq:vector-orthogonal-decomposition} into the right singular vector space $\mathbf{V}$, \ie $\mathbf{w}^{(\ell - 1)} = \sum_{j \geq 2} \alpha_j^{(\ell - 1)} \mathbf{v}_j$, where $\mathbf{v}_j$ is the $j$-th right singular vector of the unperturbed component, and $\sum_{j \geq 2} (\alpha_j^{(\ell - 1)})^2 = 1$. Then, we have:
\begin{equation}
    \mathbf{\Sigma}^{(\ell)} \mathbf{V}^\top \mathbf{w}^{(\ell - 1)} = \sum_{j \geq 2} \alpha_j^{(\ell - 1)} \sigma_j^{(\ell)} \mathbf{e}_j,
\end{equation}
where $\mathbf{e}_j$ is the $j$-th standard basis vector. Consequently, we can derive a lower bound for $\cos \theta_{\text{ideal}}^{(\ell)}$ as follows:
\begin{equation}
\begin{aligned}
    \cos \theta_{\text{ideal}}^{(\ell)}
    & = \frac{\langle \mathbf{U} \mathbf{\Sigma}^{(\ell)} \mathbf{V}^\top \mathbf{v}^{(\ell - 1)}, \mathbf{v}_1 \rangle}{\|\mathbf{U} \mathbf{\Sigma}^{(\ell)} \mathbf{V}^\top \mathbf{v}^{(\ell - 1)}\|_2} \\
    & = \frac{\sigma_1^{(\ell)} \cos \theta^{(\ell - 1)}}{\sqrt{(\sigma_1^{(\ell)})^2 \cos^2 \theta^{(\ell - 1)} + \sum_{j \geq 2} (\alpha_j^{(\ell - 1)})^2 (\sigma_j^{(\ell)})^2 \sin^2 \theta^{(\ell - 1)}}} \\
    & \geq \frac{\sigma_1^{(\ell)} \cos \theta^{(\ell - 1)}}{\sqrt{(\sigma_1^{(\ell)})^2 \cos^2 \theta^{(\ell - 1)} + (\sigma_2^{(\ell)})^2 \sin^2 \theta^{(\ell - 1)}}},
\end{aligned}
\end{equation}
where the last inequality follows from the facts that $\sum_{j \geq 2} (\alpha_j^{(\ell - 1)})^2 = 1$ and $\sigma_j^{(\ell)} \leq \sigma_2^{(\ell)}$ for $j \geq 2$. Notably, this bound can be rearranged into a much more concise form:
\begin{equation}
    \tan \theta_{\text{ideal}}^{(\ell)} \leq \frac{\sigma_2^{(\ell)}}{\sigma_1^{(\ell)}} \tan \theta^{(\ell - 1)}.
\end{equation}
This reveals that, in the absence of perturbations, the angle $\theta_{\text{ideal}}^{(\ell)}$ between the projection direction and the principal singular vector is strictly compressed by a factor of $\sigma_2^{(\ell)} / \sigma_1^{(\ell)} < 1$ relative to the previous angle $\theta^{(\ell - 1)}$ after a single layer of attention aggregation.

\noindentparnoline{Incorporating Perturbations.}
Next, we consider the actual angle $\theta^{(\ell)}$ under perturbations, which can be lower-bounded as follows:
\begin{equation}
\begin{aligned} \label{eq:angle-perturbation}
    \cos \theta^{(\ell)} 
    & = \frac{\langle \mathbf{U} \mathbf{\Sigma}^{(\ell)} \mathbf{V}^\top \mathbf{v}^{(\ell - 1)} + \mathbf{\Delta}^{(\ell)} \mathbf{v}^{(\ell - 1)}, \mathbf{v}_1 \rangle}{\|\mathbf{U} \mathbf{\Sigma}^{(\ell)} \mathbf{V}^\top \mathbf{v}^{(\ell - 1)} + \mathbf{\Delta}^{(\ell)} \mathbf{v}^{(\ell - 1)}\|_2} \\
    & \geq \frac{\sigma_1^{(\ell)} \cos \theta^{(\ell - 1)} - \epsilon}{\sqrt{(\sigma_1^{(\ell)})^2 \cos^2 \theta^{(\ell - 1)} + (\sigma_2^{(\ell)})^2 \sin^2 \theta^{(\ell - 1)}} + \epsilon}.
\end{aligned}
\end{equation}
Recall the Taylor expansion for a sufficiently small $\epsilon$:
\begin{equation}
    \frac{a - \epsilon}{b + \epsilon} = \frac{a}{b} - \frac{a + b}{b^2} \epsilon + o(\epsilon),
\end{equation}
which holds for any $a, b > 0$ and $\epsilon < |b|$. In our case, setting $a = \sigma_1^{(\ell)} \cos \theta^{(\ell - 1)}$ and $b = \sqrt{(\sigma_1^{(\ell)})^2 \cos^2 \theta^{(\ell - 1)} + (\sigma_2^{(\ell)})^2 \sin^2 \theta^{(\ell - 1)}}$, we have $(a + b) / b^2 \leq 2 / b \leq 2 / \sigma_2^{(\ell)}$ for a sufficiently small $\epsilon$. Therefore, we can simplify the bound in \eqnref{eq:angle-perturbation} to the following form:
\begin{equation}
    \cos \theta^{(\ell)} \geq \frac{\sigma_1^{(\ell)} \cos \theta^{(\ell - 1)}}{\sqrt{(\sigma_1^{(\ell)})^2 \cos^2 \theta^{(\ell - 1)} + (\sigma_2^{(\ell)})^2 \sin^2 \theta^{(\ell - 1)}}} - \frac{2\epsilon}{\sigma_2^{(\ell)}} + o(\epsilon).
\end{equation}
By applying the Taylor expansion once more, this inequality can be reformulated as:
\begin{equation} \label{eq:angle-perturbation-tan-recursive-inequality-origin}
    \tan \theta^{(\ell)} - \frac{4\epsilon}{\sigma_2^{(\ell)}} \sec^3 \theta^{(\ell)} \leq \frac{\sigma_2^{(\ell)}}{\sigma_1^{(\ell)}} \tan \theta^{(\ell - 1)} + o(\epsilon).
\end{equation}
To solve this recursive inequality, we proceed by induction on the layer index $\ell$ with the inductive hypothesis that $\tan \theta^{(\ell)} = O(\epsilon \ell / \min\{\delta, \eta\})$. For the base case $\ell = 1$, \cref{lemma:approximate-svd-perturbation-sin} establishes that $\sin \theta^{(1)} \leq 3 \epsilon / \delta$ (noting that $3 > 2 + \epsilon / \sigma_1^{(1)}$). Therefore, $\tan \theta^{(1)} \leq C \epsilon / \delta + o(\epsilon^2) = O(\epsilon / \delta)$, which satisfies the hypothesis. For the inductive step, assuming $\tan \theta^{(\ell - 1)} = O(\epsilon (\ell - 1) / \min\{\delta, \eta\})$, we have $\sec^3 \theta^{(\ell - 1)} = 1 + O(\epsilon^2 (\ell - 1)^2 / \min\{\delta, \eta\}^2)$. Thus, the inequality in \eqnref{eq:angle-perturbation-tan-recursive-inequality-origin} can be simplified to:
\begin{equation}
    \tan \theta^{(\ell)} \leq \frac{\sigma_2^{(\ell)}}{\sigma_1^{(\ell)}} \tan \theta^{(\ell - 1)} + \frac{4 \epsilon}{\sigma_2^{(\ell)}} + o(\epsilon).
\end{equation}
Applying this inequality recursively from layer $1$ to layer $\ell$ yields the following bound for the angle at the current layer:
\begin{equation} \label{eq:angle-recursive-bound}
    \tan \theta^{(\ell)} \leq \left( \prod_{\ell' = 1}^{\ell} \frac{\sigma_2^{(\ell')}}{\sigma_1^{(\ell')}} \right) \tan \theta^{(0)} + 4 \sum_{\ell' = 1}^{\ell} \frac{\epsilon}{\sigma_2^{(\ell')}} + o(\epsilon).
\end{equation}
Note that $\tan \theta^{(0)} = \tan \theta(\mathbf{v}_1, \mathbf{v}_1) = 0$, meaning the first term in \eqnref{eq:angle-recursive-bound} vanishes. The remaining terms can be further upper-bounded by $4 \epsilon \ell / \eta + o(\epsilon) = O(\epsilon \ell / \eta)$, where $\eta$ is the assumed lower bound of $\sigma_2^{(\ell)}$ across all layers. Therefore, the induction hypothesis holds for layer $\ell$, completing the induction and proving that for all layers $\ell$:
\begin{equation} \label{eq:angle-tan-bound}
    \tan \theta^{(\ell)} \leq \frac{4 \epsilon \ell}{\min\{\delta, \eta\}} + o(\epsilon).
\end{equation}

\noindentparnoline{Deriving Spectral Norm Lower Bound.}
By substituting the derived bound \eqnref{eq:angle-tan-bound} for $\tan \theta^{(\ell)}$ into the Taylor expansion of $\cos \theta^{(\ell)}$, we obtain:
\begin{equation}
    \cos \theta^{(\ell)} \geq 1 - \frac{8 \epsilon^2 \ell^2}{\min\{\delta, \eta\}^2} + o(\epsilon^2).
\end{equation}
Consequently, the lower bound of the intermediate projection norm in \eqnref{eq:projection-norm-decomposition-lower-bound} becomes:
\begin{equation}
\begin{aligned}
    \left\| \mathbf{S}^{(\ell)} \mathbf{v}^{(\ell - 1)} \right\|_2
    & \geq \sigma_1^{(\ell)} \cos \theta^{(\ell - 1)} - \epsilon \\
    & \geq \sigma_1^{(\ell)} - \epsilon - O\left( \frac{\epsilon^2 (\ell - 1)^2 \sigma_1^{(\ell)}}{\min\{\delta, \eta\}^2} \right).
\end{aligned}
\end{equation}
Finally, substituting this intermediate bound into the projection norm decomposition \eqnref{eq:projection-norm-decomposition}, we obtain the desired lower bound for the spectral norm:
\begin{equation}
    \left\| \mathbf{S}^{(L)} \mathbf{S}^{(L - 1)} \cdots \mathbf{S}^{(1)} \right\|_2
    \geq \prod_{\ell = 1}^{L} \left( \sigma_1^{(\ell)} - \epsilon - O\left( \frac{\epsilon^2 \ell^2 \sigma_1^{(\ell)}}{\min\{\delta, \eta\}^2} \right) \right).
\end{equation}
For a sufficiently small $\epsilon$, the right-hand side grows exponentially with the number of layers $L$. This completes the proof.
\end{proof}

\begin{table}[t]
    \centering
    \caption{Cosine similarity (\%) between the principal left and right singular vectors of attention score matrices, estimated on the SASRec++ backbone with different numbers of layers.}
    \label{tab:singular-vector-alignment}
    \begin{tabular}{c|ccccc}
    \toprule
    \textbf{\#Layers} & \textbf{ML-1M} & \textbf{Beauty} & \textbf{Toy} & \textbf{Electronic} \\
    \midrule
    \textbf{2}   
        & 91.83\%  & 82.99\%  & 82.58\%  & 87.30\% \\
    \textbf{4}     
        & 90.28\%  & 82.02\%  & 82.63\%  & 79.48\% \\ 
    \textbf{6}      
        & 88.00\%  & 81.76\%  & 80.06\%  & 88.21\% \\
    \textbf{8} 
        & 87.88\%  & 82.07\%  & 82.84\%  & 82.57\% \\
    \bottomrule
    \end{tabular}
\end{table}



\subsection{Justification of Assumption in \cref{theorem:spectral-norm-attention-aggregation}} \label{app:theoretical-proofs:justification-of-assumption-spectral-norm-attention-aggregation}

In \cref{theorem:spectral-norm-attention-aggregation}, we assume that the left and right principal singular vectors $\mathbf{u}_1$ and $\mathbf{v}_1$ of the attention score matrix $\mathbf{S}$ are highly aligned (omitting the layer index for brevity). Empirically, we verify this assumption by assessing the cosine similarity $\cos \theta(\mathbf{u}_1, \mathbf{v}_1)$ across different layers. As illustrated in \cref{tab:singular-vector-alignment}, this consistently reveals a high alignment, with $\cos \theta(\mathbf{u}_1, \mathbf{v}_1)$ ranging between 0.8 and 0.9.

Beyond empirical verification, we can also provide a theoretical justification for this assumption in the specific yet representative case, where each diagonal block $\mathbf{S}_u$ of the attention score matrix $\mathbf{S}$ in each layer takes the form of a Cesàro matrix~\citep{brown1965cesaro,ross2022cesaro}, representing a common autoregressive attention pattern in practice. Formally, a matrix $\mathbf{M} \in \mathbb{R}^{m \times m}$ is defined as a Cesàro matrix if:
\begin{equation}
    \mathbf{M}_{i, j} = \begin{cases}
        \frac{1}{i}, & \text{if } j \leq i, \\
        0, & \text{otherwise}.
    \end{cases}
\end{equation}
In the context of attention mechanisms, a Cesàro matrix can be interpreted as an attention score matrix where each token attends uniformly to all preceding tokens, which naturally leads to higher attention scores for popular items in sequential recommendation scenarios. For such a matrix $\mathbf{M}$, we establish the following theorem:

\begin{tcolorbox}[kessokustyle=bocchipink]
\begin{restatable}{theorem}{RestatableSingularVectorAlignmentUniformAttention} \label{theorem:singular-vector-alignment-uniform-attention}
    \upshape
    Let $\mathbf{M} \in \mathbb{R}^{m \times m}$ be a Cesàro matrix, and let $\mathbf{u}_1$ and $\mathbf{v}_1$ denote the principal left and right singular vectors of $\mathbf{M}$, respectively. Then, we have:
    \begin{equation}
        \lim_{m \to \infty} \cos \theta(\mathbf{u}_1, \mathbf{v}_1) = 1,
    \end{equation}
    where $\theta(\cdot, \cdot)$ denotes the angle between two vectors.
\end{restatable}
\end{tcolorbox}

\begin{proof}[Proof of \cref{theorem:singular-vector-alignment-uniform-attention}]
To calculate the cosine similarity between the principal left and right singular vectors $\mathbf{u}_1$ and $\mathbf{v}_1$ of the Cesàro matrix $\mathbf{M}$, we first express $\mathbf{u}_1$ in terms of $\mathbf{v}_1$.

By the definition of the SVD, we have $\mathbf{M}^\top \mathbf{u}_1 = \sigma_1 \mathbf{v}_1$, where $\sigma_1$ is the largest singular value of $\mathbf{M}$. Provided $\mathbf{M}$ is invertible, we can write $\mathbf{u}_1 = \sigma_1 (\mathbf{M}^{-1})^\top \mathbf{v}_1$, where the inverse $\mathbf{M}^{-1}$ can be computed explicitly as follows:
\begin{equation}
    \mathbf{M}^{-1}_{i, j} = \begin{cases}
    i, & j = i, \\
    -(i-1), & j = i-1, \\
    0, & \text{otherwise}.
    \end{cases}
\end{equation}
The cosine similarity $\cos \theta(\mathbf{u}_1, \mathbf{v}_1) = \mathbf{v}_1^\top \mathbf{u}_1$ can then be formulated as a quadratic form of $\mathbf{v}_1$:
\begin{equation}
    \frac{1}{\sigma_1} \cos \theta(\mathbf{u}_1, \mathbf{v}_1) = \mathbf{v}_1^\top (\mathbf{M}^{-1})^\top \mathbf{v}_1.
\end{equation}
For simplicity, let $v_i$ denote the $i$-th element of $\mathbf{v}_1$. We can expand this quadratic form as:
\begin{equation} \label{eq:cosine-quadratic-form-expansion}
    \frac{1}{\sigma_1} \cos \theta(\mathbf{u}_1, \mathbf{v}_1) = \sum_{i=1}^{m-1} i(v_i - v_{i+1})v_i + m v_m^2.
\end{equation}
Applying the basic algebraic identity $a^2 - ab = \frac{1}{2}a^2 - \frac{1}{2}b^2 + \frac{1}{2}(a-b)^2$ to the terms inside the summation, we get:
\begin{equation}
    i(v_i - v_{i+1})v_i = \frac{i}{2} v_i^2 - \frac{i}{2} v_{i+1}^2 + \frac{i}{2} (v_i - v_{i+1})^2.
\end{equation}
Substituting this back, the summation in \eqnref{eq:cosine-quadratic-form-expansion} can be rearranged into a telescoping series:
\begin{equation}
\begin{aligned}
    \frac{1}{\sigma_1} \cos \theta(\mathbf{u}_1, \mathbf{v}_1) 
    & = \sum_{i=1}^{m-1} \left( \frac{i}{2} v_i^2 - \frac{i}{2} v_{i+1}^2 \right) + \sum_{i=1}^{m-1} \frac{i}{2} (v_i - v_{i+1})^2 + m v_m^2 \\
    & = \frac{1}{2} \sum_{i=1}^{m-1} v_i^2 + \frac{m+1}{2} v_m^2 + \frac{1}{2} \sum_{i=1}^{m-1} i (v_i - v_{i+1})^2 \\
    & = \frac{1}{2} \sum_{i=1}^{m} v_i^2 + \frac{m}{2} v_m^2 + \frac{1}{2} \sum_{i=1}^{m-1} i (v_i - v_{i+1})^2 \\
    & \geq \frac{1}{2} \|\mathbf{v}_1\|_2^2 = \frac{1}{2}.
\end{aligned}
\end{equation}
From this, we directly obtain a lower bound for the cosine similarity:
\begin{equation}
    \cos \theta(\mathbf{u}_1, \mathbf{v}_1) \geq \frac{\sigma_1}{2}.
\end{equation}

Finally, it has been established that the largest singular value $\sigma_1$ of the Cesàro matrix $\mathbf{M}$ approaches 2 as $m \to \infty$ (\cf Theorem 5.1 in \citep{ross2022cesaro}). Therefore, taking the limit yields:
\begin{equation}
    \lim_{m \to \infty} \cos \theta(\mathbf{u}_1, \mathbf{v}_1) = 1,
\end{equation}
which completes the proof.
\end{proof}

Building upon \cref{theorem:singular-vector-alignment-uniform-attention}, we further generalize this result to the sparse block-diagonal case, where the global attention score matrix $\mathbf{S}$ is block-diagonal with each block being a Cesàro matrix. This directly mirrors the practical scenario in sequential recommendation, where the interactions of different users are processed independently. By combining \cref{theorem:singular-vector-alignment-uniform-attention} and the following \cref{corollary:singular-vector-alignment-block-diagonal}, we can theoretically justify our assumption in \cref{theorem:spectral-norm-attention-aggregation} that the principal left and right singular vectors of the attention score matrix are highly aligned across different layers in practice.

\begin{tcolorbox}[kessokustyle=ryoblue]
\begin{restatable}{corollary}{RestatableSingularVectorAlignmentBlockDiagonal} \label{corollary:singular-vector-alignment-block-diagonal}
    \upshape
    Let $\mathbf{S} = \operatorname{diag}(\mathbf{S}_u)_{u \in \mathcal{U}} \in \mathbb{R}^{N \times N}$ be a block-diagonal attention matrix, where each block $\mathbf{S}_u \in \mathbb{R}^{n_u \times n_u}$ is a Cesàro matrix representing the causal uniform attention of user $u$. Let $\mathbf{u}_1$ and $\mathbf{v}_1$ be any pair of corresponding principal left and right singular vectors of the global matrix $\mathbf{S}$. Then, we have:
    \begin{equation} \label{eq:global-cosine-bound}
        \cos \theta(\mathbf{u}_1, \mathbf{v}_1) \geq \frac{\sigma_1(\mathbf{S})}{2}.
    \end{equation}
    Furthermore, as the maximum sequence length $n_{\max} = \max_{u \in \mathcal{U}} n_u \to \infty$, we have $\lim_{n_{\max} \to \infty} \cos \theta(\mathbf{u}_1, \mathbf{v}_1) = 1$.
\end{restatable}
\end{tcolorbox}

\begin{proof}[Proof of \cref{corollary:singular-vector-alignment-block-diagonal}]

For any user $u$, let $\mathbf{u}_i^{(u)} \in \mathbb{R}^{n_u}$ and $\mathbf{v}_i^{(u)} \in \mathbb{R}^{n_u}$ denote the $i$-th left and right singular vectors of the block $\mathbf{S}_u$, respectively, for $i = 1, 2, \cdots, n_u$. The corresponding singular values are denoted as $\sigma_i(\mathbf{S}_u)$. We pad these local singular vectors with zeros to match the global dimension $N$ while maintaining their original coordinates within the block-diagonal structure, resulting in $\tilde{\mathbf{u}}_i^{(u)} \in \mathbb{R}^N$ and $\tilde{\mathbf{v}}_i^{(u)} \in \mathbb{R}^N$. Clearly, $\tilde{\mathbf{u}}_i^{(u)}$ and $\tilde{\mathbf{v}}_i^{(u)}$ are the left and right singular vectors of the global matrix $\mathbf{S}$ corresponding to the same singular value $\sigma_i(\mathbf{S}_u)$. Furthermore, the union of all block singular values $\{\sigma_i(\mathbf{S}_u) : u \in \mathcal{U}, i = 1, 2, \cdots, n_u\}$ constitutes the entire singular value spectrum of $\mathbf{S}$. Therefore, the global maximum singular value $\sigma_1(\mathbf{S})$ must be achieved by at least one of the block singular values, \ie $\sigma_1(\mathbf{S}) = \max_{u \in \mathcal{U}} \sigma_1(\mathbf{S}_u)$.

Since each block $\mathbf{S}_u$ is a Cesàro matrix, we can apply \cref{theorem:singular-vector-alignment-uniform-attention} to obtain the following local bound for each user block:
\begin{equation} \label{eq:local-cosine-bound}
    \cos \theta(\tilde{\mathbf{u}}_1^{(u)}, \tilde{\mathbf{v}}_1^{(u)}) = \cos \theta(\mathbf{u}_1^{(u)}, \mathbf{v}_1^{(u)}) \geq \frac{\sigma_1(\mathbf{S}_u)}{2}.
\end{equation}
Therefore, if there is a unique dominant block that achieves the global maximum singular value $\sigma_1(\mathbf{S})$, then the corresponding zero-padded principal singular vectors of this block are also global, and their cosine similarity is directly lower-bounded by $\sigma_1(\mathbf{S}) / 2$, establishing the desired bound in \eqnref{eq:global-cosine-bound}.

In the more general case where multiple dominant blocks achieve the same global maximum singular value, we can still establish the identical lower bound. Specifically, assume there are $r \geq 1$ dominant blocks that achieve the exact same global maximum singular value $\sigma_1(\mathbf{S})$. The principal singular space is therefore an $r$-dimensional degenerate subspace spanned by the zero-padded principal singular vectors $\{ \tilde{\mathbf{u}}^{(j)} \}_{j = 1}^{r}$ and $\{ \tilde{\mathbf{v}}^{(j)} \}_{j = 1}^{r}$ of these dominant blocks. Any global principal right singular vector $\mathbf{v}_1$ can be expressed as an orthogonal linear combination of the zero-padded principal right singular vectors $\tilde{\mathbf{v}}^{(j)}$ corresponding to these dominant blocks, \ie
\begin{equation}
    \mathbf{v}_1 = \sum_{j=1}^r \alpha_j \tilde{\mathbf{v}}^{(j)}, \quad \text{\st } \sum_{j=1}^r \alpha_j^2 = 1.
\end{equation}
Correspondingly, the left singular vector is strictly constrained by the exact same linear coefficients:
\begin{equation}
    \mathbf{u}_1 = \frac{1}{\sigma_1(\mathbf{S})} \mathbf{S} \mathbf{v}_1 = \sum_{j=1}^r \alpha_j \tilde{\mathbf{u}}^{(j)}.
\end{equation}
The cosine similarity can then be calculated as:
\begin{equation}
    \cos \theta(\mathbf{u}_1, \mathbf{v}_1) = \left\langle \sum_{i=1}^r \alpha_i \tilde{\mathbf{u}}^{(i)}, \sum_{j=1}^r \alpha_j \tilde{\mathbf{v}}^{(j)} \right\rangle =  \sum_{j=1}^r \alpha_j^2 \langle \tilde{\mathbf{u}}^{(j)}, \tilde{\mathbf{v}}^{(j)} \rangle.
\end{equation}
Applying the local bound \eqnref{eq:local-cosine-bound} to these dominant blocks, we obtain:
\begin{equation}
    \cos \theta(\mathbf{u}_1, \mathbf{v}_1) \geq \sum_{j=1}^r \alpha_j^2 \left( \frac{\sigma_1(\mathbf{S})}{2} \right) = \frac{\sigma_1(\mathbf{S})}{2},
\end{equation}
which exactly matches the bound \eqnref{eq:global-cosine-bound} obtained in the unique dominant block case.

Finally, as the maximum sequence length $n_{\max} \to \infty$, the global maximum singular value $\sigma_1(\mathbf{S}) = \max_u \sigma_1(\mathbf{S}_u) \to 2$. Substituting this limit into our bound \eqnref{eq:global-cosine-bound} yields $\lim_{n_{\max} \to \infty} \cos \theta(\mathbf{u}_1, \mathbf{v}_1) \ge 1$. Given that cosine similarity is strictly capped at 1, we conclude that $\lim_{n_{\max} \to \infty} \cos \theta(\mathbf{u}_1, \mathbf{v}_1) = 1$, which completes the proof.
\end{proof}



\subsection{Proof of \cref{theorem:spectral-norm-bound}} \label{app:theoretical-proofs:proof-of-spectral-norm-bound}

To prove \cref{theorem:spectral-norm-bound}, we first extend our proposed objective in \eqnref{eq:overall-objective} to general transformer architectures, which may include multi-head attention, residual connections, layer normalization, and non-linear activation functions. Notably, our core idea is to establish an upper bound for the spectral norm of the model predictions (\ie $\| \hat{\mathbf{Y}} \|_2$) in a backward manner, recursively applying the spectral norm bounds from the output of each layer back to the input.

\noindentparnoline{Prediction Layer.}
We first consider the prediction layer, where the model output $\hat{\mathbf{Y}}$ is generated by the multiplication between the last layer output $\mathbf{X}^{(L)}$ and the item embedding matrix $\mathbf{E}$, \ie $\hat{\mathbf{Y}} = \mathbf{X}^{(L)} \mathbf{E}^\top$. Therefore, we can derive the following spectral norm bound for the final predictions:
\begin{equation}
    \| \hat{\mathbf{Y}} \|_2 \leq \| \mathbf{X}^{(L)} \|_2 \cdot \| \mathbf{E} \|_2,
\end{equation}
where the first term $\| \mathbf{X}^{(L)} \|_2$ is the spectral norm of the last layer output, which can be recursively upper-bounded by the spectral norm bounds of each preceding layer. The second term $\| \mathbf{E} \|_2$ can be regularized analogously to the proposed feed-forward regularization in \eqnref{eq:feed-forward-regularization-derivation} by the power iteration method and can be viewed as a part of the feed-forward regularization.

\noindentparnoline{Residual Connections.}
Next, we consider residual connections, which are typically added after each attention or feed-forward layer. Formally, this operation can be bounded by:
\begin{equation}
\begin{aligned}
    \| \mathbf{X}_{\text{block-res}} \|_2
    & = \| \mathbf{X}_{\text{block}} + \mathbf{X}_{\text{block-input}} \|_2 \\
    & \leq \| \mathbf{X}_{\text{block}} \|_2 + \| \mathbf{X}_{\text{block-input}} \|_2,
\end{aligned}
\end{equation}
where $\mathbf{X}_{\text{block}}$ is the output of the attention or feed-forward block, and $\mathbf{X}_{\text{block-input}}$ is the input to this block. If we can bound the spectral norm of $\mathbf{X}_{\text{block}}$ by $\| \mathbf{X}_{\text{block}} \|_2 \leq C_\text{block} \cdot \| \mathbf{X}_{\text{block-input}} \|_2$ for some scalar amplification factor $C_\text{block} > 0$, then we can bound the spectral norm of the residual output $\mathbf{X}_{\text{block-res}}$ by:
\begin{equation}
    \| \mathbf{X}_{\text{block-res}} \|_2 \leq (1 + C_\text{block}) \cdot \| \mathbf{X}_{\text{block-input}} \|_2.
\end{equation}
Thus, the residual connection can simply be viewed as introducing a multiplicative factor of $(1 + C_\text{block})$, as opposed to $C_\text{block}$ in the non-residual case presented in the main text. Consequently, residual connections do not alter the overall form of the spectral norm bound. Our goal then reduces to bounding the spectral norm amplification factor $C_\text{block}$ of the individual attention and feed-forward blocks, which is precisely achieved by our proposed regularizations.

\begin{table}[t]
    \centering
    \caption{Dataset statistics.}
    \label{tab:dataset-statistics}
    \begin{tabular}{l|ccccc}
    \toprule
    \textbf{Dataset} & \textbf{\#Users} & \textbf{\#Items} & \textbf{\#Interactions} & \textbf{Density} \\
    \midrule
    \textbf{ML-20M}
        & 138,493 & 18,345  & 19,984,024       & 0.0787\% \\
    \textbf{ML-1M}   
        & 6,040   & 3,416   & 999,611         & 4.8448\%   \\
    \textbf{Beauty}     
        & 22,332  & 12,086   & 198,215        & 0.0734\%   \\
    \textbf{Toy}      
        & 19,124  & 11,758  & 165,247        & 0.0735\%   \\
    \textbf{Electronic} 
        & 28,566  & 15,743  & 241,198        & 0.0536\%   \\
    \textbf{Clothing} 
        & 39,230  & 22,948  & 277,534        & 0.0308\%   \\
    \textbf{Book} 
        & 38,234  & 38,519  & 517,167        & 0.0351\%   \\
    \bottomrule
    \end{tabular}
\end{table}

\noindentparnoline{Layer Normalization.}
Layer normalization is typically applied prior to the attention and feed-forward blocks (\ie pre-norm). In this work, we specifically consider RMSNorm~\citep{dubey2024llama}, which generally exhibits better performance than the original LayerNorm with a bias term. Formally, the RMSNorm operation is defined as follows:
\begin{equation}
    \mathbf{X}_{\text{block-norm}} = \mathbf{X}_{\text{block-input-RMS}} \times \operatorname{diag}(\mathbf{\gamma}_{\text{block}}),
\end{equation}
where $\mathbf{X}_{\text{block-input-RMS}}$ is the root-mean-square normalized representation obtained from the original input $\mathbf{X}_{\text{block-input}}$, and $\mathbf{\gamma}_{\text{block}}$ is the learnable scaling parameter. To ensure numerical stability, we assume that the minimum 2-norm of the row vectors in any $\mathbf{X}_{\text{block-input}}$ is lower-bounded by a positive constant $\xi_r$, and the maximum absolute value of any element in $\mathbf{\gamma}_{\text{block}}$ is upper-bounded by a positive constant $\xi_\gamma$. We can then derive the following spectral norm bound for the layer normalization operation:
\begin{equation}
    \| \mathbf{X}_{\text{block-norm}} \|_2 \leq \frac{\xi_\gamma}{\xi_r} \| \mathbf{X}_{\text{block-input}} \|_2.
\end{equation}

\noindentparnoline{Feed-forward Network.}
Following the standard transformer architecture, we consider a two-layer feed-forward network with a non-linear activation function $\sigma(\cdot)$, \ie $\mathbf{X}_{\text{ffn}} = \sigma(\mathbf{X}_{\text{ffn-norm}} \mathbf{W}_1) \mathbf{W}_2$. We assume that the activation function $\sigma(\cdot)$ is Lipschitz continuous with a finite Lipschitz constant $L_\sigma$, and $\sigma(\mathbf{0}) = \mathbf{0}$. For instance, the ReLU activation satisfies this assumption with $L_\sigma = 1$. Then, we establish the spectral norm bound for the feed-forward block:
\begin{equation} \label{eq:feed-forward-spectral-norm-bound}
\begin{aligned}
    \| \mathbf{X}_{\text{ffn}} \|_2
    & \leq \| \sigma(\mathbf{X}_{\text{ffn-norm}} \mathbf{W}_1) \|_2 \cdot \| \mathbf{W}_2 \|_2 \\
    & \leq L_\sigma \sqrt{d} \| \mathbf{X}_{\text{ffn-norm}} \mathbf{W}_1 \|_2 \cdot \| \mathbf{W}_2 \|_2 \\
    & \leq L_\sigma \sqrt{d} \| \mathbf{W}_1 \|_2 \| \mathbf{W}_2 \|_2 \cdot \| \mathbf{X}_{\text{ffn-norm}} \|_2,
\end{aligned}
\end{equation} 
where the second inequality relies on the spectral norm bound for element-wise operations, \ie $\| \sigma(\cdot) \|_2 \leq \| \sigma(\cdot) \|_F \leq L_\sigma \| \cdot \|_F \leq L_\sigma \sqrt{d} \| \cdot \|_2$. The input term $\mathbf{X}_{\text{ffn-norm}}$ can be recursively upper-bounded. Note that the coefficient $L_\sigma \sqrt{d}$ is a constant for a given model architecture, which can be absorbed into the aforementioned $\frac{\xi_\gamma}{\xi_r}$ term from the layer normalization. Consequently, the entire feed-forward block possesses the following spectral norm amplification factor $C_\text{ffn}$:
\begin{equation} \label{eq:feed-forward-regularization-c-factor}
    C_\text{ffn} = \frac{\xi_\gamma}{\xi_r} L_\sigma \sqrt{d} \| \mathbf{W}_1 \|_2 \| \mathbf{W}_2 \|_2,
\end{equation}
which perfectly aligns with our proposed feed-forward regularization in \eqnref{eq:feed-forward-regularization-derivation} up to a constant multiplier.

\noindentparnoline{Single-head Attention.}
In the main text, we derived the spectral norm bound \eqnref{eq:attention-regularization} for the single-head attention case. Specifically, during the attention aggregation of a given layer, \ie $\mathbf{X}_{\text{attn}} = \mathbf{S} \mathbf{X}_{\text{attn-norm}} \mathbf{W}_V \mathbf{W}_O$, we have:
\begin{equation} \label{eq:attention-spectral-norm-bound-single-head}
    \| \mathbf{X}_{\text{attn}} \|_2 \leq \| \mathbf{S} \|_2 \cdot \| \mathbf{W}_V \|_2 \| \mathbf{W}_O \|_2 \cdot \| \mathbf{X}_{\text{attn-norm}} \|_2.
\end{equation}
Here, the constituent terms are bounded as follows:
\begin{itemize}[topsep=3pt,leftmargin=10pt,itemsep=0pt]
    \item The first term $\| \mathbf{S} \|_2$ can be upper-bounded by the 1-norm $\sqrt{\| \mathbf{S} \|_1}$ as shown in \eqnref{eq:attention-regularization-derivation}, which is further bounded by our proposed attention regularizer $\sqrt{\widetilde{\| \mathbf{S} \|}_{1}}$ defined in \eqnref{eq:attention-regularization}.
    \item The second term $\| \mathbf{W}_V \|_2 \| \mathbf{W}_O \|_2$ is directly governed by the proposed feed-forward regularization in \eqnref{eq:feed-forward-regularization-derivation}.
    \item The last term $\| \mathbf{X}_{\text{attn-norm}} \|_2$ can be recursively upper-bounded.
\end{itemize}
Similarly, the entire single-head attention block with pre-norm introduces the following spectral norm amplification factor $C_\text{attn-single}$:
\begin{equation} \label{eq:attention-regularization-c-factor}
    C_\text{attn-single} = \frac{\xi_\gamma}{\xi_r} \sqrt{\widetilde{\| \mathbf{S} \|}_{1}} \cdot \| \mathbf{W}_V \|_2 \| \mathbf{W}_O \|_2,
\end{equation}
which seamlessly decomposes into the exact attention and feed-forward regularizers proposed in the main text.

\noindentparnoline{Multi-head Attention.}
For multi-head attention with $H$ heads, the input $\mathbf{X}_{\text{attn-norm}}$ is first projected into $H$ value matrices via projection weights $\mathbf{W}_{V, h} \in \mathbb{R}^{d \times d/H}$, \ie $\mathbf{V}_{h} = \mathbf{X}_{\text{attn-norm}} \mathbf{W}_{V, h}$ for $h = 1, 2, \cdots, H$. The attention aggregation is then performed separately for each head:
\begin{equation}
    \mathbf{X}_{\text{attn}} = [\mathbf{S}_1 \mathbf{V}_1, \mathbf{S}_2 \mathbf{V}_2, \cdots, \mathbf{S}_H \mathbf{V}_H] \mathbf{W}_O.
\end{equation}
Denoting $\mathbf{V}_h' = \mathbf{S}_h \mathbf{V}_h$ as the aggregated value matrix for head $h$, we utilize the following inequality for matrix concatenation:
\begin{equation}
    \| [\mathbf{V}_1', \mathbf{V}_2', \cdots, \mathbf{V}_H'] \|_2 \leq \sqrt{\| \mathbf{V}_1' \|_2^2 + \| \mathbf{V}_2' \|_2^2 + \cdots + \| \mathbf{V}_H' \|_2^2}.
\end{equation}
Specifically, letting $\mathbf{V}' = [\mathbf{V}_1', \mathbf{V}_2', \cdots, \mathbf{V}_H']$ represent the concatenated matrix, this inequality is derived as follows:
\begin{equation}
\begin{aligned}
    \| \mathbf{V}' \|_2^2 
    & = \| \mathbf{V}' (\mathbf{V}')^\top \|_2
    = \left\| \sum_{h = 1}^{H} \mathbf{V}_h' (\mathbf{V}_h')^\top \right\|_2 \\
    & \leq \sum_{h = 1}^{H} \| \mathbf{V}_h' (\mathbf{V}_h')^\top \|_2
    = \sum_{h = 1}^{H} \| \mathbf{V}_h' \|_2^2.
\end{aligned}
\end{equation}
Recall that for each head $h$, we established in \eqnref{eq:attention-spectral-norm-bound-single-head} that $\| \mathbf{V}_h' \|_2 \leq \sqrt{\widetilde{\| \mathbf{S}_h \|}_{1}} \cdot \| \mathbf{W}_{V, h} \|_2 \cdot \| \mathbf{X}_{\text{attn-norm}} \|_2$. Thus, we can bound the spectral norm for the multi-head attention operation as:
\begin{equation}
    \| \mathbf{X}_{\text{attn}} \|_2 \leq \sqrt{\sum_{h = 1}^{H} \widetilde{\| \mathbf{S}_h \|}_{1} \cdot \| \mathbf{W}_{V, h} \|_2^2} \cdot \| \mathbf{W}_O \|_2 \cdot \| \mathbf{X}_{\text{attn-norm}} \|_2,
\end{equation}
which yields the multi-head spectral norm amplification factor $C_\text{attn-multi}$:
\begin{equation} \label{eq:attention-regularization-c-factor-multi-head}
    C_\text{attn-multi} = \frac{\xi_\gamma}{\xi_r} \sqrt{\sum_{h = 1}^{H} \widetilde{\| \mathbf{S}_h \|}_{1} \cdot \| \mathbf{W}_{V, h} \|_2^2} \cdot \| \mathbf{W}_O \|_2.
\end{equation}
The bound in \eqnref{eq:attention-regularization-c-factor-multi-head} exhibits a slightly more complex form than the single-head case in \eqnref{eq:attention-regularization-c-factor}. The trailing term $\| \mathbf{W}_O \|_2$ naturally falls under the feed-forward regularization in \eqnref{eq:feed-forward-regularization-derivation}, while the leading square-root term acts as a generalized form of the proposed attention regularization in \eqnref{eq:attention-regularization}. Although we adopt the compact regularization form derived from \eqnref{eq:attention-regularization-c-factor-multi-head} in our practical implementations, we empirically find that optimizing the regularizers separately for each head (\ie $\sqrt{\widetilde{\| \mathbf{S}_h \|}_{1}} \cdot \| \mathbf{W}_{V, h} \|_2$) performs similarly well. This is expected, as minimizing the individual components within a sum intrinsically drives down the overall summation.

\begin{figure}[t]
    \centering
    \includegraphics[width=\columnwidth]{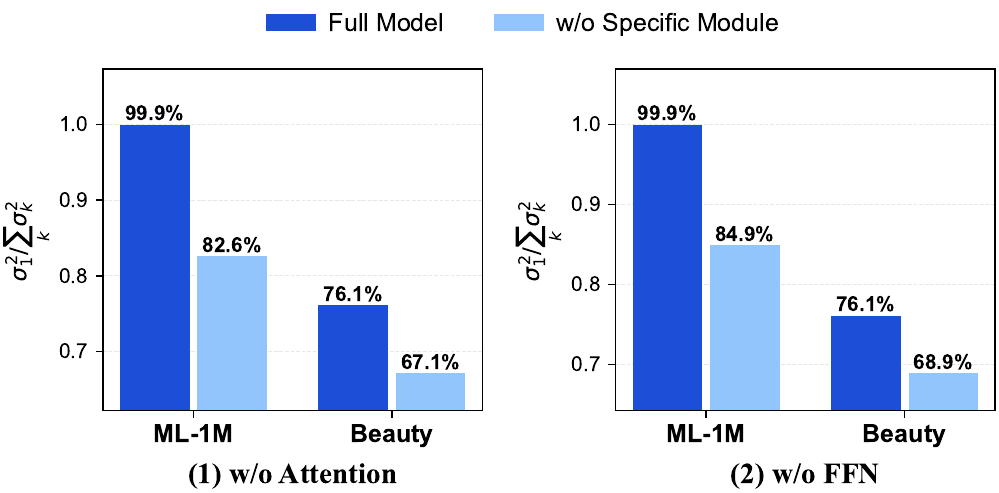}
    \caption{
        Ablation study investigating the impact of attention and feed-forward components on the spectral collapse of model predictions, evaluated using an 8-layer SASRec++ backbone with 256-dim embeddings. We compare the severity of spectral collapse in the full transformer model against variants where (1) 6 attention blocks or (2) 6 feed-forward blocks are removed. The results demonstrate that removing either component effectively mitigates the spectral collapse. Consequently, both components inherently drive the exacerbation of this collapse, verifying the necessity of jointly regularizing both architectural components to restrain popularity bias amplification.
    }
    \Description{Ablation: the effect of attention and FFN on spectral norm.}
    \label{fig:ablation-spectral-norm}
\end{figure}

Now we are ready to prove \cref{theorem:spectral-norm-bound}:

\begin{tcolorbox}[kessokustyle=bocchipink]
\RestateTheoremSpectralNormBound*
\end{tcolorbox}

\begin{proof}[Proof of \cref{theorem:spectral-norm-bound}]
Throughout the comprehensive derivations in \cref{app:theoretical-proofs:proof-of-spectral-norm-bound}, we have systematically bounded the spectral norm of each block's output relative to its input using a closed-form amplification factor $C_\text{block}$ ($+1$ if a residual connection is applied), as defined in \eqnref{eq:feed-forward-regularization-c-factor}, \eqnref{eq:attention-regularization-c-factor}, and \eqnref{eq:attention-regularization-c-factor-multi-head} for feed-forward, single-head attention, and multi-head attention blocks, respectively. Therefore, by recursively propagating these bounds from the output layer back to the input layer, we establish the overall spectral norm bound for the model predictions $\hat{\mathbf{Y}}$:
\begin{equation}
    \| \hat{\mathbf{Y}} \|_2
    \leq \| \mathbf{X}^{(L)} \|_2 \cdot \| \mathbf{E} \|_2
    \leq \prod_{\text{blocks}} C_\text{block} \cdot \| \mathbf{X}^{(0)} \|_2.
\end{equation}
Here, the item embedding term $\| \mathbf{E} \|_2$ is regularized identically to the feed-forward weights, and the input term $\| \mathbf{X}^{(0)} \|_2$, which consists of selected vectors from the item embedding matrix, is trivially bounded by $\| \mathbf{E} \|_2$ as well. Consequently, taking the logarithm of both sides directly recovers the exact form of our proposed regularization objective in \eqnref{eq:overall-objective}, completing the proof of \cref{theorem:spectral-norm-bound}.
\end{proof}



\begin{table}[t]
    \centering
    \caption{
        Architecture configurations for large-scale scaling experiments on the ML-20M dataset. Here, $L$ is the number of layers, $d$ is the hidden dimensionality, and $H$ is the number of attention heads. The last column "\#Params" denotes the approximate number of non-embedding parameters. This is calculated as $12Ld^2$, reflecting the standard transformer architecture excluding the embedding layer.
    }
    \label{tab:large-scale-architecture}
    \begin{tabular}{l|ccc|c}
    \toprule
    \textbf{Model Size} & $L$ & $d$ & $H$ & \textbf{\#Params} \\
    \midrule
    XXS        & 1              & 64            & 1             & 0.05M                           \\
    XS         & 2              & 128           & 2             & 0.39M                           \\
    Small      & 4              & 256           & 4             & 3.10M                            \\
    Base       & 6              & 512           & 8             & 18.9M                           \\
    Large      & 8              & 768           & 12            & 56.6M                           \\
    XL         & 10             & 1024          & 16            & 0.13B                \\
    XXL        & 12             & 1536          & 24            & 0.34B                 \\
    \bottomrule
    \end{tabular}
\end{table}


\section{Experimental Supplementary} \label{app:experimental-supplementary}

\subsection{Datasets} \label{app:experimental-supplementary:datasets}

To ensure a fair comparison with existing methods, we conduct experiments on six widely used benchmark datasets from MovieLens~\citep{harper2015movielens} and Amazon~\citep{he2016ups,mcauley2015image}, specifically including ML-1M/20M, Beauty, Toy, Electronic, Clothing, and Book. The dataset preprocessing strictly follows the standard procedure from previous studies~\citep{yang2024psl,yang2025breaking,zhang2026talos}, where users and items with fewer than 5 interactions are filtered out (\ie 5-core), and the remaining interactions are sorted in chronological order. For each user, the most recent interaction is used for testing, the second most recent interaction is used for validation, and all remaining historical interactions are used for training~\citep{kang2018self,yang2026bear}. Specifically, following the common practice in sequential recommendation~\citep{kang2018self}, the historical interaction sequences for each user are truncated to a maximum length of 200 on ML-1M/20M and 50 on the Amazon datasets. Detailed statistics of the datasets are summarized in \cref{tab:dataset-statistics}.


\subsection{Backbones} \label{app:experimental-supplementary:backbone-details}

We evaluate each method on two advanced transformer-based sequential recommendation backbones, namely SASRec++~\citep{zhang2024scaling,guo2024scaling} (2024) and HSTU~\citep{zhai2024actions} (2024). The former adapts SASRec~\citep{kang2018self} with several architectural adjustments for better scalability, while the latter modifies the attention mechanism to increase model capacity and improve computational efficiency. In our main experiments, we scale the number of layers up to 8 and the embedding dimensionality up to 256. These scale settings deliberately exceed the typical configurations used in prior studies, allowing us to better assess the scalability of different methods. Furthermore, we conduct large-scale experiments on the ML-20M dataset using SASRec++ with up to 0.34B non-embedding parameters (12 layers and 1536-dimensional embeddings), as illustrated in \cref{fig:ml-20m-scaling-laws}. The backbone details are provided below:

\noindentparnoline{SASRec++~\citep{zhang2024scaling,guo2024scaling}.}
\citet{kang2018self} initially proposed the SASRec architecture, which consists of a stack of transformer blocks. As model sizes scaled up, \citet{zhang2024scaling} and \citet{guo2024scaling} identified several architectural bottlenecks in the original design and proposed a series of adjustments to enhance scalability. Following these works, we adopt the improved SASRec++ architecture as one of our evaluation backbones. SASRec++ directly utilizes the standard LlamaDecoder architecture~\citep{dubey2024llama}, incorporating pre-norm LayerNorm, RoPE~\citep{su2024roformer} for positional encoding, an intermediate dimensionality of $4d$ for the feed-forward network, and optional dropout on the attention weights. In the main experiments, the number of attention heads is searched over $\{1, 2, 4, 8\}$ for different layer configurations, while it is fixed to $d / 64$ for the large-scale experiments on ML-20M. Specifically, for the layer scaling experiments, we fix the embedding dimensionality to 512 and vary the number of layers. Detailed architectural configurations for the large-scale experiments are provided in \cref{tab:large-scale-architecture}.

\noindentparnoline{HSTU~\citep{zhai2024actions}.}
HSTU is a recently proposed backbone that achieves state-of-the-art performance in industrial-scale sequential recommendation. The core modification of HSTU is the introduction of a SiLU-based attention mechanism, which replaces the traditional softmax function with the SiLU activation to compute attention scores. In this work, we generalize the original HSTU architecture to support an optional feed-forward network (which is absent in the original HSTU) and an optional gated attention mechanism (which is present in the original HSTU). Empirically, we observe that in certain settings, incorporating the feed-forward network and removing the gated attention can further improve training stability and performance. To ensure a fair comparison, all experiments on the HSTU backbone are conducted by searching over these feed-forward and gated attention options, and we report the best results for each method.

\noindentparnoline{Training Details.}
For all evaluated methods, we adopt the Binary Cross-Entropy (BCE) loss~\citep{kang2018self} with negative sampling as the main recommendation loss in \eqnref{eq:overall-objective}, where the number of negative samples is set to 32 for ML-1M/20M and 16 for Amazon datasets. We utilize the AdamW optimizer~\citep{loshchilov2017decoupled} with a learning rate of $10^{-3}$ and a weight decay of $0.1$ in our main experiments. The warmup ratio is set to $0.05$, after which the learning rate is linearly decayed to zero. For the large-scale experiments on the ML-20M dataset, we search over learning rates in $\{10^{-3}, 5 \times 10^{-4}, 10^{-4}\}$ and weight decays in $\{0.1, 0.01\}$ to ensure stable training. The batch size is fixed at 512 for main experiments and 1024 for large-scale experiments, and the maximum number of training epochs is set to 200 across all experiments. The optimal results for all methods in the main experiments are selected based on the NDCG@5 metric. The remaining hyperparameters for the compared methods and \mymethod are searched over the ranges specified in \cref{app:experimental-supplementary:compared-methods}.


\begin{figure}[t]
    \centering
    \includegraphics[width=\columnwidth]{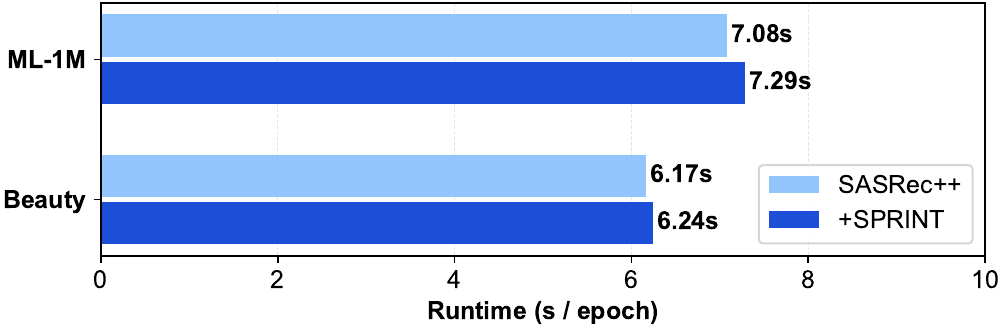}
    \caption{
        Runtime comparison (seconds per epoch) between the standard SASRec++ backbone and +\mymethod. The results demonstrate that the proposed regularizations introduce only a negligible training overhead.
    }
    \Description{Runtime comparison.}
    \label{fig:runtime-comparison}
\end{figure}

\subsection{Compared Methods} \label{app:experimental-supplementary:compared-methods}

The details of the compared methods are summarized below:

\begin{itemize}[topsep=3pt,leftmargin=10pt,itemsep=0pt]
    \item \textbf{DROS}~\citep{yang2023generic}: A distributionally robust optimization (DRO) method that incorporates KL-divergence-based DRO~\citep{wang2024distributionally} regularization, weighted by the item popularity distribution, to enhance model robustness against distribution shifts. Key hyperparameters: temperature $\beta_0 \in \{0.5, 1.0, 1.5\}$, weight $\alpha \in \{0.1, 0.3, 0.5\}$.
    \item \textbf{MOJITO}~\citep{tran2023attention}: A temporal-aware disentanglement method that separates temporal context and item representations, explicitly disentangling the short-term time-dependent intents from context-independent long-term preferences. Key hyperparameters: weight $\lambda \in \{0.01, 0.1, 0.5\}$.
    \item \textbf{TPAB}~\citep{yoo2025generalizable}: A temporal-aware disentanglement approach that separates intrinsic item property embeddings from temporal-popularity embeddings via exponential time bucketing. Key hyperparameters: number of buckets $K \in \{20, 40, 60\}$, bootstrapping weight $\lambda \in \{0, 0.5, 1.0, 1.5, 2.0, 2.5, 3.0\}$.
    \item \textbf{R$^2$Rec}~\citep{li2025reembedding}: A propensity reweighting method that incorporates adaptive loss reweighting, coupled with a Gaussian reembedding operation for tail items, to explicitly prioritize tail item optimization. Key hyperparameter: temperature $\tau \in \{0.05, 0.1, 0.5, 1.0\}$.
    \item \textbf{D$^2$LR}~\citep{lu2025dual}: A propensity reweighting method that performs token-level inverse propensity score (IPS) debiasing. Key hyperparameters: $\alpha \in \{0.1, 0.2, \dots, 0.9\}$ and $\beta \in \{0.05, 0.1, \dots, 0.4\}$.
    \item \textbf{LogDet}~\citep{zhang2023mitigating}: A spectral regularization method that mitigates dimensional collapse by regularizing the covariance matrix of item embeddings via log-determinant divergence. Key hyperparameter: weight $\lambda \in \{10^{-3}, 10^{-4}, 10^{-5}, 10^{-6}\}$.
    \item \textbf{ReSN}~\citep{lin2025recommendation}: A spectral regularization method that constrains the largest singular value of the predicted interaction matrix to mitigate popularity bias. Key hyperparameter: regularization weight $\lambda \in \{10^{-3}, 10^{-4}, 10^{-5}, 10^{-6}\}$.
    \item \textbf{\mymethod (Ours)}: The proposed method that jointly regularizes the spectral norms of attention and feed-forward components to mitigate the spectral collapse during scaling. Key hyperparameters: attention regularization weight $\lambda_{\text{attn}} \in \{2.5, 5.0, 7.5, 10\}$ for the main experiments, and $\lambda_{\text{attn}} \in \{0.05, 0.1, 0.25, 0.5\}$ for the large-scale experiments on the ML-20M dataset, feed-forward regularization weight $\lambda_{\text{ffn}} \in \{10^{-1}, 10^{-2}, 10^{-3}, 10^{-4}\}$.
\end{itemize}


\subsection{Supplementary Results} \label{app:experimental-supplementary:supplementary-results}

In addition to the results reported in the main text, we further present supplementary results to provide a more comprehensive evaluation of \mymethod and the compared methods.

\noindentparnoline{Popularity Memorization Effect in Sequential Recommendation.}
\cref{lemma:popularity-bias-spectral-norm} states that the popularity information within recommendation data is theoretically encoded in the principal right singular vector $\mathbf{q}_1$ of the model predictions $\hat{\mathbf{Y}}$. In sequential recommendation, we also empirically observe this phenomenon: the Spearman's rank correlation between item popularity $\mathbf{r}$ and the principal right singular vector $\mathbf{q}_1$ is remarkably high (\eg 0.9038 on the ML-1M dataset, and 0.8258 on the Toy dataset). This strong correlation empirically verifies the existence of popularity memorization in sequential recommendation.

\begin{table*}[t]
  \centering
  \caption{Overall performance comparison on SASRec++ and HSTU. The columns "N@$K$", "H@$K$", and "Fair" denote the NDCG@$K$, HitRatio@$K$, and Fair-0.8 metrics, respectively. The best results are highlighted in bold, and the strongest baselines are underlined. "\textcolor{darkred}{\textbf{Imp.\%}}" indicates the relative improvement of \mymethod over the best baseline (\cf \cref{tab:main-results-5} for top-5 results).}
  \resizebox{\textwidth}{!}{
    \begin{tabular}{l|ccc|ccc|ccc|ccc|ccc|ccc}
    \toprule
    \multicolumn{1}{c|}{\multirow{2}[4]{*}{\textbf{Methods}}} & \multicolumn{3}{c|}{\textbf{ML-1M}} & \multicolumn{3}{c|}{\textbf{Beauty}} & \multicolumn{3}{c|}{\textbf{Toy}} & \multicolumn{3}{c|}{\textbf{Electronic}} & \multicolumn{3}{c|}{\textbf{Clothing}} & \multicolumn{3}{c}{\textbf{Book}} \\
\cmidrule(lr){2-4} \cmidrule(lr){5-7} \cmidrule(lr){8-10} \cmidrule(lr){11-13} \cmidrule(lr){14-16} \cmidrule(lr){17-19}
    & \textbf{N@10} & \textbf{H@10} & \textbf{Fair} & \textbf{N@10} & \textbf{H@10} & \textbf{Fair} & \textbf{N@10} & \textbf{H@10} & \textbf{Fair} & \textbf{N@10} & \textbf{H@10} & \textbf{Fair} & \textbf{N@10} & \textbf{H@10} & \textbf{Fair} & \textbf{N@10} & \textbf{H@10} & \textbf{Fair} \\
    \midrule
    \textbf{SASRec} & .1015 & .1938 & .2227 & .0315 & .0615 & .2694 & .0392 & \uline{.0780} & .3741 & .0148 & .0277 & .0029 & .0026 & .0060 & .3806 & .0301 & .0664 & .4128 \\
    +DROS & .1045 & .1989 & .2172 & .0319 & .0619 & .2592 & .0351 & .0694 & .3478 & .0149 & .0282 & .0035 & .0027 & .0061 & \uline{.4157} & .0300 & .0663 & .4154 \\
    +MOJITO & .1055 & .2028 & .2188 & .0350 & .0605 & \uline{.3202} & \uline{.0420} & .0710 & .3803 & .0157 & \uline{.0310} & .0002 & .0030 & .0051 & .3717 & .0322 & .0538 & .2375 \\
    +TPAB & .1045 & .2031 & .1921 & .0281 & .0537 & .2291 & .0291 & .0529 & .3011 & .0072 & .0147 & .0004 & .0005 & .0010 & .3631 & .0302 & .0531 & .2156 \\
    +R$^2$Rec & .1008 & .2002 & .1997 & .0327 & .0658 & .3108 & .0368 & .0706 & \uline{.3878} & \uline{.0166} & .0304 & .0043 & .0025 & .0052 & .3500 & \uline{.0323} & .0654 & .3028 \\
    +D$^2$LR & .1051 & .1970 & \uline{.2317} & .0323 & .0670 & .3152 & .0359 & .0709 & .3543 & .0142 & .0263 & .0031 & .0023 & .0048 & .2994 & .0312 & \uline{.0674} & \uline{.4208} \\
    +LogDet & \uline{.1075} & \uline{.2035} & .1792 & .0349 & .0670 & .1071 & .0399 & .0743 & .2833 & .0160 & \uline{.0310} & .0006 & .0031 & .0058 & .1528 & .0316 & .0611 & .1832 \\
    +ReSN & .0987 & .2030 & .1223 & \uline{.0352} & \uline{.0674} & .1016 & .0409 & .0723 & .1515 & .0157 & .0288 & \uline{.0175} & \uline{.0038} & \uline{.0069} & .0850 & .0302 & .0453 & .3572 \\
\cmidrule(lr){1-1} \cmidrule(lr){2-4} \cmidrule(lr){5-7} \cmidrule(lr){8-10} \cmidrule(lr){11-13} \cmidrule(lr){14-16} \cmidrule(lr){17-19}
    \textbf{+\mymethod} & \textbf{.1155} & \textbf{.2121} & \textbf{.2534} & \textbf{.0406} & \textbf{.0692} & \textbf{.4127} & \textbf{.0485} & \textbf{.0804} & \textbf{.4836} & \textbf{.0194} & \textbf{.0336} & \textbf{.1618} & \textbf{.0043} & \textbf{.0075} & \textbf{.4454} & \textbf{.0389} & \textbf{.0709} & \textbf{.4953} \\
    \textcolor{darkred}{\textbf{Imp.\%}} & \textcolor{darkred}{\textbf{7.45\%}} & \textcolor{darkred}{\textbf{4.26\%}} & \textcolor{darkred}{\textbf{2.17\%}} & \textcolor{darkred}{\textbf{15.26\%}} & \textcolor{darkred}{\textbf{2.62\%}} & \textcolor{darkred}{\textbf{9.25\%}} & \textcolor{darkred}{\textbf{15.55\%}} & \textcolor{darkred}{\textbf{3.13\%}} & \textcolor{darkred}{\textbf{9.58\%}} & \textcolor{darkred}{\textbf{17.37\%}} & \textcolor{darkred}{\textbf{8.40\%}} & \textcolor{darkred}{\textbf{14.43\%}} & \textcolor{darkred}{\textbf{10.74\%}} & \textcolor{darkred}{\textbf{9.33\%}} & \textcolor{darkred}{\textbf{2.97\%}} & \textcolor{darkred}{\textbf{20.53\%}} & \textcolor{darkred}{\textbf{5.21\%}} & \textcolor{darkred}{\textbf{7.45\%}} \\
    \midrule
    \textbf{HSTU} & .1037 & .2036 & .1902 & .0341 & .0647 & .3369 & .0379 & .0686 & .4215 & .0194 & .0358 & .0014 & .0030 & \uline{.0059} & .4179 & .0327 & .0572 & .1815 \\
    +DROS & .1036 & .1990 & .1934 & .0346 & \uline{.0666} & \uline{.3861} & .0375 & .0643 & \uline{.4951} & .0196 & .0359 & .0009 & .0030 & .0058 & \uline{.4582} & .0332 & .0572 & .1739 \\
    +MOJITO & .1037 & .2047 & .2019 & .0343 & .0626 & .3174 & .0367 & .0593 & .3753 & .0208 & .0388 & .0001 & .0031 & .0056 & .3856 & .0344 & .0602 & .1768 \\
    +TPAB & .1001 & .1997 & .2240 & .0349 & .0639 & .3781 & .0355 & .0621 & .4133 & .0135 & .0250 & .0001 & .0011 & .0021 & .3834 & .0358 & \uline{.0641} & .1383 \\
    +R$^2$Rec & .0949 & .1936 & .2128 & .0327 & .0645 & .3269 & .0382 & .0692 & .4429 & .0201 & .0364 & .0021 & .0030 & \uline{.0059} & .3430 & .0338 & .0622 & .1679 \\
    +D$^2$LR & .1042 & .2038 & \uline{.2305} & .0335 & .0648 & .3824 & .0378 & .0696 & .4589 & .0198 & .0372 & .0084 & .0029 & \uline{.0059} & .4508 & .0351 & .0603 & \uline{.2490} \\
    +LogDet & \uline{.1075} & .2072 & .1574 & \uline{.0350} & .0640 & .3110 & \uline{.0387} & \uline{.0709} & .4425 & .0216 & \uline{.0395} & \uline{.0116} & \uline{.0033} & .0057 & .3155 & \uline{.0360} & \uline{.0641} & .1029 \\
    +ReSN & .1018 & \uline{.2086} & .1520 & .0335 & \uline{.0666} & .1398 & .0346 & .0615 & .2350 & \uline{.0219} & .0358 & .0000 & .0032 & \uline{.0059} & .1042 & .0250 & .0363 & .0752 \\
\cmidrule(lr){1-1} \cmidrule(lr){2-4} \cmidrule(lr){5-7} \cmidrule(lr){8-10} \cmidrule(lr){11-13} \cmidrule(lr){14-16} \cmidrule(lr){17-19}
    \textbf{+\mymethod} & \textbf{.1109} & \textbf{.2097} & \textbf{.2503} & \textbf{.0413} & \textbf{.0695} & \textbf{.4449} & \textbf{.0482} & \textbf{.0724} & \textbf{.5173} & \textbf{.0244} & \textbf{.0420} & \textbf{.2150} & \textbf{.0041} & \textbf{.0070} & \textbf{.4744} & \textbf{.0389} & \textbf{.0678} & \textbf{.3237} \\
    \textcolor{darkred}{\textbf{Imp.\%}} & \textcolor{darkred}{\textbf{3.10\%}} & \textcolor{darkred}{\textbf{0.57\%}} & \textcolor{darkred}{\textbf{1.99\%}} & \textcolor{darkred}{\textbf{18.30\%}} & \textcolor{darkred}{\textbf{4.26\%}} & \textcolor{darkred}{\textbf{5.89\%}} & \textcolor{darkred}{\textbf{24.57\%}} & \textcolor{darkred}{\textbf{2.22\%}} & \textcolor{darkred}{\textbf{2.22\%}} & \textcolor{darkred}{\textbf{11.36\%}} & \textcolor{darkred}{\textbf{6.22\%}} & \textcolor{darkred}{\textbf{20.34\%}} & \textcolor{darkred}{\textbf{24.68\%}} & \textcolor{darkred}{\textbf{17.32\%}} & \textcolor{darkred}{\textbf{1.62\%}} & \textcolor{darkred}{\textbf{7.88\%}} & \textcolor{darkred}{\textbf{5.72\%}} & \textcolor{darkred}{\textbf{7.47\%}} \\
    \bottomrule
    \end{tabular}
  }
  \label{tab:main-results-10}
\end{table*}

\noindentparnoline{Ablation Study on Transformer Components and Spectral Collapse.}
In \cref{fig:ablation-spectral-norm}, we conduct an ablation study to isolate the effects of the attention and feed-forward components on the spectral collapse of model predictions. Specifically, we compare the severity of spectral collapse in the full transformer model against variants where multiple attention blocks or feed-forward blocks are omitted. The results clearly indicate that both components independently contribute to the exacerbation of spectral collapse. Replacing either component with identity mappings (\ie removing them) significantly mitigates this collapse. This finding empirically verifies the necessity of jointly regularizing both the attention and feed-forward components to effectively control the spectral collapse inherent in scaling sequential recommenders.

\noindentparnoline{Cross-Layer Alignment of Singular Vectors.}
In \cref{theorem:spectral-norm-attention-aggregation}, we assume that the singular vectors of the attention score matrices are highly aligned across layers, up to small perturbations. This crucial assumption enables us to derive a closed-form lower bound for the spectral norm of multi-layer attention aggregation. To empirically verify this, we compute the cosine similarity between the singular vectors of the attention score matrices across different layers within a 8-layer trained transformer. Our results reveal that this cross-layer cosine similarity remains consistently high (above 0.7788 on the ML-1M dataset, and above 0.7845 on the Beauty dataset). This evidence empirically verifies the validity of our assumption and confirms the practical applicability of the derived spectral norm bound. Refer to \cref{app:theoretical-proofs:justification-of-assumption-spectral-norm-attention-aggregation} for further theoretical and empirical justifications of the assumption in \cref{theorem:spectral-norm-attention-aggregation}.

\begin{table}[t]
  \centering
  \caption{Ablation study on \mymethod regularizations.}
  \resizebox{\columnwidth}{!}{
    \begin{tabular}{l|cc|cc}
    \toprule
    \multicolumn{1}{c|}{\multirow{2}[4]{*}{\textbf{Backbones}}} & \multicolumn{2}{c|}{\textbf{ML-1M}} & \multicolumn{2}{c}{\textbf{Beauty}} \\
\cmidrule(lr){2-3} \cmidrule(lr){4-5}
    & \textbf{NDCG@5} & \textbf{Fairness} & \textbf{NDCG@5} & \textbf{Fair-0.8} \\
    \midrule
    $\lambda_{\text{attn}} = \lambda_{\text{ffn}} = 0$ & 0.0807 & 0.2227 & 0.0248 & 0.2694 \\
    $\lambda_{\text{attn}} = 0$ & 0.0819 & 0.2512 & 0.0310 & 0.3428 \\
    $\lambda_{\text{ffn}} = 0$ & 0.0850 & 0.2388 & 0.0308 & 0.3672 \\
    \textbf{+\mymethod} & \textbf{0.0925} & \textbf{0.2534} & \textbf{0.0355} & \textbf{0.4127} \\
    \bottomrule
    \end{tabular}
  }
  \label{tab:ablation-study}
\end{table}

\noindentparnoline{Empirical Runtime Analysis.}
As discussed in \cref{sec:methodology:analyses_and_discussions}, the proposed regularizers are theoretically lightweight. To empirically validate this, we record the actual training time per epoch across different datasets. As detailed in \cref{fig:runtime-comparison}, implementing \mymethod on top of the standard SASRec++ backbone incurs merely a marginal computational overhead (approximately 3\%). This empirical evidence confirms that our method seamlessly scales to practical recommender systems without sacrificing training efficiency.

\noindentparnoline{(RQ1) Overall Performance.}
\cref{tab:main-results-5,tab:main-results-10} presents the overall recommendation performance comparison of \mymethod and compared methods on the SASRec++ and HSTU backbones, where the best results are reported based on NDCG@5 metric (\cf \cref{sec:experiments:results} for further discussions). In addition to the Fair@1-0.8 metric~\citep{lin2025recommendation} (\ie long-tail exposure in top-1 recommendations), we also report multiple other fairness metrics, including Fair@5-0.8, Fair@10-0.8, Gini@5~\citep{klimashevskaia2024survey}, ARP@5 (Average Recommendation Popularity)~\citep{klimashevskaia2024survey}, and BQS@5 (Balanced Quality Score, based on NDCG@5)~\citep{coppolillo2024balanced}, to comprehensively evaluate the fairness performance of different models. As illustrated in \cref{tab:fairness-results}, \mymethod consistently outperforms all compared methods across different fairness metrics, demonstrating its superior ability to mitigate popularity bias.

\noindentparnoline{(RQ2) Layer Scaling Performance.}
To examine the effectiveness of \mymethod in deep transformers, we conduct layer scaling experiments in \cref{fig:layer-scaling-toy-electronic,fig:layer-scaling-ml1m-beauty} (\cf \cref{sec:experiments:results} for further discussions).

\noindentparnoline{(RQ3) Ablation Study and Hyperparameter Analysis.} We perform an ablation study in \cref{tab:ablation-study} to isolate the contributions of attention and feed-forward regularizations in \mymethod. The sensitivity of \mymethod to the regularization weights $\lambda_{\text{attn}}$ and $\lambda_{\text{ffn}}$ is further analyzed in \cref{fig:hyper-parameter-sensitivity} (\cf \cref{sec:experiments:results} for further discussions).

\begin{table*}[p]
  \centering
  \caption{Fairness performance comparison on Beauty and Toy datasets. The columns "ARP", "Gini", "BQS", and "Fair" denote the Average Recommendation Popularity@5, Gini Index@5, Balanced Quality Score@5, and Fair-0.8@$K$ ($K \in \{1, 5, 10\}$) metrics, respectively. The best results are highlighted in bold. For ARP$\downarrow$ and Gini$\downarrow$, lower is better; for BQS$\uparrow$ and Fair$\uparrow$, higher is better.}
  \resizebox{\textwidth}{!}{
    \begin{tabular}{l|cccccc|cccccc}
    \toprule
    \multicolumn{1}{c|}{\multirow{2}[4]{*}{\textbf{Methods}}} & \multicolumn{6}{c|}{\textbf{Beauty}} & \multicolumn{6}{c}{\textbf{Toy}} \\
\cmidrule(lr){2-7} \cmidrule(lr){8-13}
    & \textbf{ARP@5}$\downarrow$ & \textbf{Gini@5}$\downarrow$ & \textbf{BQS@5}$\uparrow$ & \textbf{Fair@1}$\uparrow$ & \textbf{Fair@5}$\uparrow$ & \textbf{Fair@10}$\uparrow$ & \textbf{ARP@5}$\downarrow$ & \textbf{Gini@5}$\downarrow$ & \textbf{BQS@5}$\uparrow$ & \textbf{Fair@1}$\uparrow$ & \textbf{Fair@5}$\uparrow$ & \textbf{Fair@10}$\uparrow$ \\
    \midrule
    \textbf{SASRec++} & 52.4 & 0.76 & 0.5000 & 0.2694 & 0.3193 & 0.3186 & 29.8 & 0.63 & 0.5000 & 0.3741 & 0.4189 & 0.4192 \\
    +LogDet & 61.7 & 0.81 & 0.5008 & 0.1071 & 0.1490 & 0.1659 & 34.3 & 0.76 & 0.4996 & 0.2833 & 0.3373 & 0.3402 \\
    +ReSN & 65.1 & 0.81 & 0.5016 & 0.1016 & 0.1911 & 0.2162 & 42.1 & 0.81 & 0.4992 & 0.1515 & 0.2237 & 0.2434 \\
\cmidrule(lr){1-1} \cmidrule(lr){2-7} \cmidrule(lr){8-13}
    \textbf{+\mymethod} & \textbf{35.7} & \textbf{0.57} & \textbf{0.5042} & \textbf{0.4127} & \textbf{0.4375} & \textbf{0.4249} & \textbf{21.7} & \textbf{0.54} & \textbf{0.5033} & \textbf{0.4836} & \textbf{0.5206} & \textbf{0.5047} \\
    \bottomrule
    \end{tabular}
  }
  \label{tab:fairness-results}
\end{table*}

\begin{figure*}[p]
    \centering
    \includegraphics[width=0.95\textwidth]{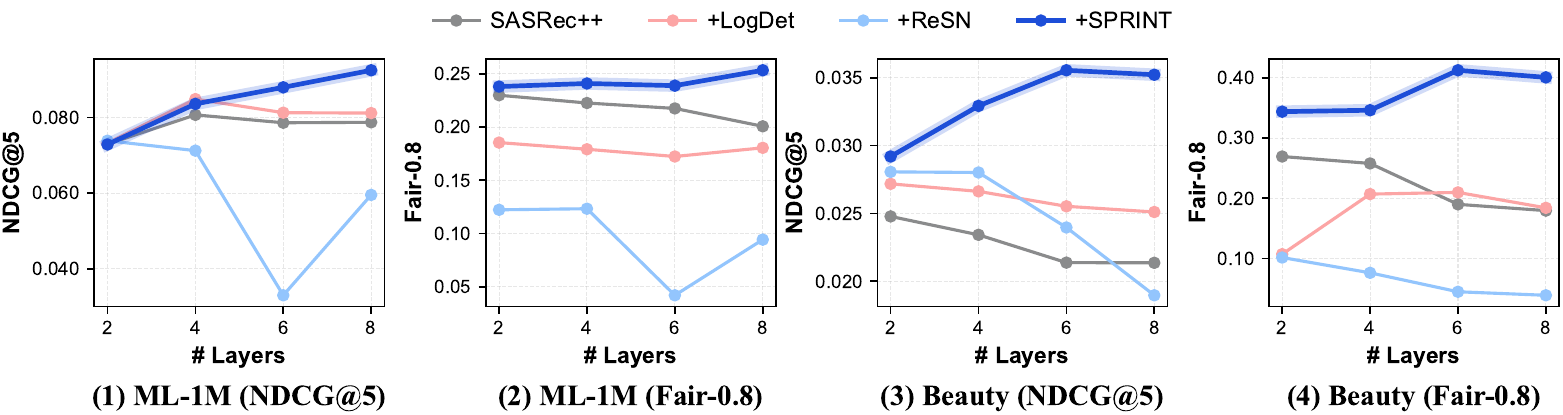}
    \caption{
        Layer scaling results of \mymethod and representative baselines on the SASRec++ backbone, where the number of layers is varied from 2 to 8, and the accuracy and fairness metrics are reported. Refer to \cref{fig:layer-scaling-toy-electronic} for results on other datasets.
    }
    \Description{Layer scaling results.}
    \label{fig:layer-scaling-ml1m-beauty}
\end{figure*}

\begin{figure*}[p]
    \centering
    \includegraphics[width=\textwidth]{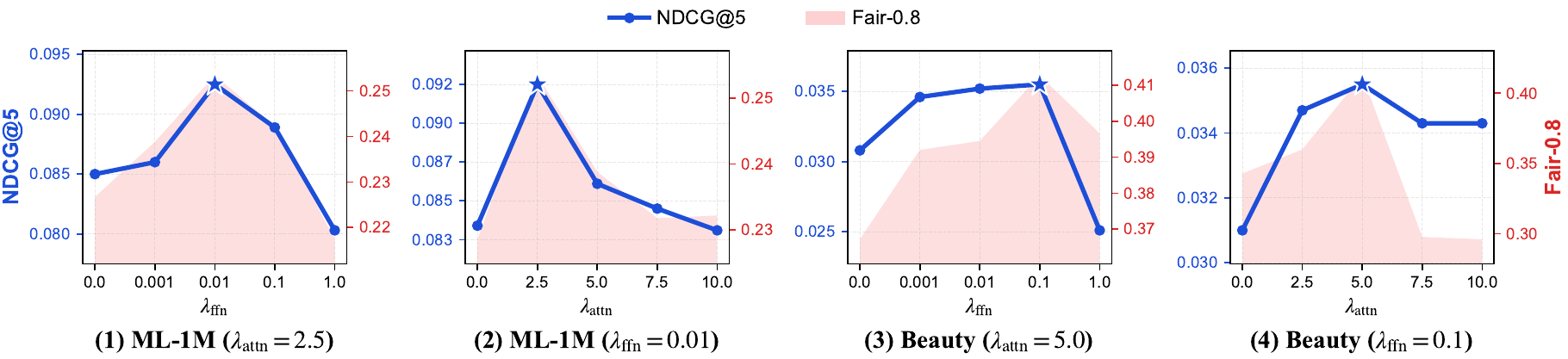}
    \caption{
        Hyperparameter sensitivity analysis of \mymethod on the SASRec++ backbone with respect to the attention regularization weight $\lambda_\text{attn}$ and the feed-forward regularization weight $\lambda_\text{ffn}$.
    }
    \Description{Hyper-parameter sensitivity analysis.}
    \label{fig:hyper-parameter-sensitivity}
\end{figure*}




\end{document}